\documentclass{IEEEtran}

\usepackage{amsmath}
\usepackage{amssymb}
\usepackage{amsthm}
\usepackage{cite}
\usepackage{graphicx}
\usepackage{verbatim}
\usepackage{bbm}
\usepackage{subfigure}
\usepackage{url}
\usepackage{mathrsfs}
\usepackage{enumerate}
\usepackage{overpic}
\usepackage{psfrag}
\usepackage{color}

\newcommand{\mb}[1]{{\bf #1}}

\renewcommand{\c}{\mb{c}}

\newcommand{\e}{\mb{e}}

\newcommand{\bv}{\mb{v}}

\newcommand{\A}{\mb{A}}

\newcommand{\B}{\mb{B}}

\newcommand{\D}{\mb{D}}

\newcommand{\R}{\mb{R}}

\newcommand{\HH}{{\mathcal H}}
\newcommand{\GG}{{\mathcal G}}
\newcommand{\I}{\mb{I}}
\newcommand{\J}{\mb{J}}
\newcommand{\KK}{{\mathcal K}}
\newcommand{\M}{\mb{M}}
\newcommand{\MM}{{\mathcal M}}

\renewcommand{\SS}{{\mathcal S}}
\newcommand{\U}{\mb{U}}
\newcommand{\UU}{{\mathcal U}}

\newcommand{\XX}{{\mathfrak X}}
\newcommand{\YY}{{\mathcal Y}}

\newcommand{\bLambda}{\mb{\Lambda}}
\newcommand{\bGamma}{\mb{\Gamma}}

\newcommand{\zero}{\mb{0}}

\newcommand{\beq}{\begin{equation}}
\newcommand{\eeq}{\end{equation}}

\newcommand{\bth}{{\boldsymbol \theta}}

\newcommand{\bdelta}{{\boldsymbol \delta}}
\newcommand{\bmu}{{\boldsymbol \mu}}

\newcommand{\eps}{\varepsilon}

\newcommand{\PG}{\mb{P}_{\GG}}

\newcommand{\PS}{\mb{P}_{\SS}}
\newcommand{\pdht}{\pd{h_\bth}{\bth}}
\newcommand{\Pdy}{P_{\bth_0}(dy)}

\newcommand{\CC}{{\mathbb C}}
\newcommand{\RR}{{\mathbb R}}
\newcommand{\ZZ}{{\mathbb Z}}

\newcommand{\MSE}{{\mathrm{MSE}}}

\newcommand{\E}[1]{{\mathbb E} \! \left\{ #1 \right\}}
\renewcommand{\Pr}[1]{\mathrm{Pr} \kern -1pt \left\{ #1 \right\}}
\newcommand{\pd}[2]{\frac{\partial #1}{\partial #2}}

\newcommand{\Nu}[1]{{{\mathcal N}\! \left( #1 \right) }}
\newcommand{\Ra}[1]{{{\mathcal R}\! \left( #1 \right) }}

\newcommand{\bra}{\left\langle}
\newcommand{\ket}{\right\rangle}

\newtheorem{theorem}{Theorem}
\newtheorem{lemma}{Lemma}

\DeclareMathOperator{\Cov}{Cov}
\DeclareMathOperator{\Tr}{Tr}

\DeclareMathOperator{\spn}{span}

\DeclareMathOperator{\sinc}{sinc}

\begin{document}

\title{Performance~Bounds and~Design~Criteria for~Estimating Finite~Rate~of~Innovation Signals}
\author{Zvika~Ben-Haim,~\IEEEmembership{Student~Member,~IEEE}, Tomer~Michaeli,~\IEEEmembership{Student~Member,~IEEE}, and~Yonina~C.~Eldar,~\IEEEmembership{Senior~Member,~IEEE}\thanks{%
The authors are with the Department of Electrical Engineering, Technion---Israel Institute of Technology, Haifa 32000, Israel (e-mail: \{zvikabh@tx, tomermic@tx, yonina@ee\}.technion.ac.il). This work was supported in part by a Magneton grant from the Israel Ministry of Industry and Trade, by the Israel Science Foundation under Grant 1081/07, and by the European Commission's FP7 Network of Excellence in Wireless COMmunications NEWCOM++ (grant agreement no.\ 216715).}}

\maketitle

\begin{abstract}
In this paper, we consider the problem of estimating finite rate of innovation (FRI) signals from noisy measurements, and specifically analyze the interaction between FRI techniques and the underlying sampling methods. We first obtain a fundamental limit on the estimation accuracy attainable regardless of the sampling method. Next, we provide a bound on the performance achievable using any specific sampling approach. Essential differences between the noisy and noise-free cases arise from this analysis. In particular, we identify settings in which noise-free recovery techniques deteriorate substantially under slight noise levels, thus quantifying the numerical instability inherent in such methods. This instability, which is only present in some families of FRI signals, is shown to be related to a specific type of structure, which can be characterized by viewing the signal model as a union of subspaces. Finally, we develop a methodology for choosing the optimal sampling kernels based on a generalization of the Karhunen--Lo\`eve transform. The results are illustrated for several types of time-delay estimation problems.
\end{abstract}


\section{Introduction}
\label{se:intro}

The field of digital signal processing hinges on the availability of techniques for sampling analog signals, thus converting them to discrete measurements. The sampling mechanism aims to preserve the information present in the analog domain, ideally permitting flawless recovery of the original signal. For example, one may wish to recover a continuous-time signal $x(t)$ from a discrete set of samples. The archetypical manifestation of this concept is the Shannon sampling theorem, which states that a $B$-bandlimited function can be reconstructed from samples taken at the Nyquist rate $2B$ \cite{shannon48}. 

Recently, considerable attention has been devoted to the extension of sampling theory to functions having a finite rate of innovation (FRI). These are signals determined by a finite number $\rho$ of parameters per time unit \cite{vetterli02}. Such a definition encompasses a rich variety of signals, including splines, shift-invariant signals, multiband signals, and pulse streams. In many FRI settings, several existing algorithms are guaranteed to recover the signal $x(t)$ from samples taken at rate $\rho$ \cite{vetterli02, maravic05, dragotti07, gedalyahu09, gedalyahu11, tur11}. In other words, signals which correspond to the FRI model can be reconstructed from samples taken at the rate of innovation, which is potentially much lower than their Nyquist rate.

The set of signals described by an FRI model can often be viewed as a union of subspaces \cite{gedalyahu09, gedalyahu11, tur11, eldar09}. For example, consider a stream of pulses parameterized by pulse locations and amplitudes. The set of all pulses having a given location is a subspace of the space of continuous-time functions. Thus, the set of all pulses having arbitrary locations is a union of such subspaces. As we will see, this point of view yields a flexible and productive framework for understanding the types of constraints implied by the model.

Real-world signals are often contaminated by continuous-time noise and thus do not conform precisely to the FRI model. Furthermore, like any mathematical model, the FRI framework is an approximation which does not precisely hold in practical scenarios, an effect known as mismodeling error \cite{eldar09}. It is therefore of interest to quantify the effect of noise and mismodeling errors on FRI techniques \cite{maravic05, berent10, tur11, gedalyahu11}. In the noisy case, it is no longer possible to perfectly recover the original signal from its samples. Nevertheless, one might hope for an appropriate finite-rate technique which achieves the best possible estimation accuracy, in the sense that increasing the sampling rate confers no further performance benefits. For example, to recover a $B$-bandlimited signal contaminated by continuous-time white noise, one can use an ideal low-pass filter with cutoff $B$ prior to sampling at a rate of $2B$. This strategy removes all noise components with frequencies larger than $B$, while leaving all signal components intact. Consequently, any alternative method which does not zero out frequencies above $B$ can be improved upon, whereas methods which zero out some of the signal frequencies can suffer from an arbitrarily large reconstruction error. Thus, sampling at a rate of $2B$ is indeed optimal in the case of a $B$-bandlimited signal, if the signal is corrupted by continuous-time noise prior to sampling. Sampling at a rate higher than $2B$ can be beneficial only when the sampling process itself introduces additional noise into the system, e.g., as a result of quantization.

By contrast, empirical observations indicate that, for some noisy FRI signals, substantial performance improvements are achievable when the sampling rate is increased beyond the rate of innovation \cite{dragotti07, tur11}. Thus, in some cases, there appears to be a fundamental difference between the noiseless and noise-corrupted settings, in terms of the required sampling rate. Our first goal in this paper will be to provide an analytical justification and quantification of these empirical findings. As we will see, the fact that oversampling improves performance is not merely indicative of flaws in existing algorithms; rather, it is a consequence of the inherent difficulty of reconstructing FRI signals under noise. Indeed, we will demonstrate that for some FRI signals, unless considerable oversampling is employed, performance will necessarily deteriorate by several orders of magnitude relative to the optimal achievable reconstruction capability. Such effects occur even when the noise level is exceedingly low. Our analysis will also enable us to identify and characterize the types of signals for which oversampling is necessary.

To demonstrate these results, we first derive the Cram\'er--Rao bound (CRB) for estimating a finite-duration segment of an FRI signal $x(t)$ directly from continuous-time measurements $y(t)=x(t)+w(t)$, where $w(t)$ is a Gaussian white noise process. This yields a lower bound on the accuracy whereby $x(t)$ can be recovered by any technique, regardless of its sampling rate. This setting is to be distinguished from previous bounds in the FRI literature \cite{stoica89, berent10} in three respects. First and most importantly, the measurements are a continuous-time process $y(t)$ and the bound therefore applies regardless of the sampling method. Second, in our model, the noise is added prior to sampling. Thus, as will be shown below, even sampling at an arbitrarily high rate will not completely compensate for the noise. Third, we bound the mean-squared error (MSE) in estimating $x(t)$ and not the parameters defining it, since we seek to determine the accuracy with which $x(t)$ itself can be recovered. Such a bound does not depend on the specific parametrization of the signal, and consequently, possesses a simpler analytical expression.

In practice, rather than processing the continuous-time signal $y(t)$, it is typically desired to estimate $x(t)$ from a discrete set of samples $\{c_n\}$ of $y(t)$. In this scenario, in addition to the continuous-time noise $w(t)$, digital noise may arise from the sampling process itself, for example due to quantization. To quantify the extent to which sampling degrades the ability to recover the signal, we next derive the CRB for estimating $x(t)$ from the measurements $\{c_n\}$. This analysis depends on the relative power of the two noise factors. When only digital noise is present, oversampling can be used to completely overcome its effect. On the other hand, when there exists only continuous-time noise, the bound converges to the continuous-time CRB as the sampling rate increases. In some cases, these bounds coincide at a finite sampling rate, which implies that the sampling scheme has captured all of the information present in the continuous-time signal, and any further increase in the sampling rate is useless. Conversely, when the continuous-time and sampled CRBs differ, the gap between these bounds is indicative of the degree to which information is lost in the sampling process. Our technique can then be used to plot the best possible performance as a function of the sampling rate, and thus provide the practitioner with a tool for evaluating the benefits of oversampling.

When a certain sampling technique achieves the performance of continuous-time measurements, it can be identified using the method described above. However, in some cases no such technique exists, or the sampling rate it requires may be prohibitive. In these cases, it is desirable to determine the optimal sampling scheme having an allowed rate. Since different signals are likely to perform successfully with different sampling kernels, a Bayesian or average-case analysis is well-suited for this problem. Specifically, we assume that the signal $x(t)$ has a known prior distribution over the class of signals, and determine the linear sampling and reconstruction technique which minimizes the MSE for recovering $x(t)$ from its measurements. While nonlinear reconstruction techniques are commonly used and typically outperform the best linear estimator, this approach provides a simple means for identifying an appropriate sampling method. The resulting method can then be used in conjunction with standard nonlinear FRI algorithms.

We demonstrate our results via the problem of estimating a finite-duration sequence of pulses having unknown positions and amplitudes \cite{vetterli02, dragotti07, gedalyahu09, tur11}. In this case, a simple sufficient condition is obtained for the existence of a sampling scheme whose performance bound coincides with the continuous-time CRB\@. This scheme is based on sampling the Fourier coefficients of the pulse shape, and is reminiscent of recent time-delay estimation algorithms \cite{tur11}. However, while the sampling scheme is theoretically sufficient for optimal recovery of $x(t)$, we show that in some cases there is room for substantial improvement in the reconstruction stage of these algorithms. Finally, we demonstrate that the Fourier domain is also optimal (in the sense of minimizing the reconstruction MSE) when the sampling budget is limited. Specifically, given an allowed number of samples $N$, the reconstruction MSE is minimized by sampling the $N$ highest-variance Fourier coefficients of the signal $x(t)$.

The rest of this paper is organized as follows. The problem setting is defined in Section~\ref{se:setting}, and some examples of signals conforming to this model are presented in Section~\ref{se:types}. We then briefly summarize our main results in Section~\ref{se:summary}. In Section~\ref{se:math}, we provide a technical generalization of the CRB to general spaces. This result is used to obtain bounds on the achievable reconstruction error from continuous-time measurements (Section~\ref{se:cont-time}) and using a sampling mechanism (Section~\ref{se:sampled}). Next, in Section~\ref{se:bayesian} a Bayesian viewpoint is introduced and utilized to determine the optimal sampling kernels having a given rate budget. The results are demonstrated for the specific signal model of time-delay estimation in Section~\ref{se:app}.

\section{Definitions}
\label{se:setting}

\subsection{Notation}

The following notation is used throughout the paper. A boldface lowercase letter $\bv$ denotes a vector, while a boldface uppercase letter $\M$ denotes a matrix. $\I_N$ is the $N \times N$ identity matrix. For a vector $\bv$, the notation $\|\bv\|$ indicates the Euclidean norm. Given a complex number $z \in \CC$, the symbols $z^*$ and $\Re\{z\}$ denote the complex conjugate and the real part of $z$, respectively. For an operator $P$, the range space and null space are $\Ra{P}$ and $\Nu{P}$, respectively, while the trace and adjoint are denoted, respectively, by $\Tr(P)$ and $P^*$. The Kronecker delta, denoted by $\delta_{m,n}$, equals $1$ when $m=n$ and $0$ otherwise. The expectation of a random variable $v$ is written as $\E{v}$.

The Hilbert space of square-integrable complex-valued functions over $[0,T_0]$ is denoted $L_2[0,T_0]$ or simply $L_2$. The corresponding inner product is
\beq
\bra f, g \ket \triangleq \int_0^{T_0} f(t) g^*(t) dt
\eeq
and the induced norm is $\|f\|^2_{L_2} \triangleq \bra f,f \ket$. For an ordered set of $K$ functions $g_1,\ldots,g_K$ in $L_2$, we define the associated \emph{set transformation} $G:\CC^K\rightarrow L_2$ as
\begin{equation}\label{eq:set}
(G\bv)(t) = \sum_{k=1}^K \bv_k g_k(t).
\end{equation}
By the definition of the adjoint, it follows that
\begin{equation}\label{eq:set_adjoint}
G^*f = (\bra f,g_1 \ket , \ldots, \bra f,g_K \ket)^T.
\end{equation}

\subsection{Setting}

In this work, we are interested in the problem of estimating FRI signals from noisy measurements. To define FRI signals formally, let the $T_0$-local number of degrees of freedom $N_{T_0}(t)$ of a signal $x(t)$ at time $t$ be the number of parameters defining the segment $\{ x(t) : t \in [t-T_0/2,t+T_0/2] \}$. The $T_0$-local rate of innovation of $x(t)$ is then defined as \cite{vetterli02}
\begin{equation} \label{eq:def rho}
\rho_{T_0} = \max_{t\in\RR} \frac{N_{T_0}(t)}{T_0}.
\end{equation}
We then say that $x(t)$ is an FRI signal if $\rho_{T_0}$ is finite for all sufficiently large values of $T_0$. In Section~\ref{se:types}, we will give several examples of FRI signals and compute their rates of innovation.

For concreteness, let us focus on the problem of estimating the finite-duration segment $\{ x(t) : t \in [0,T_0] \}$, for some constant $T_0$, and let $K \triangleq N_{T_0}(T_0/2)$ denote the number of parameters defining this segment. We then have
\beq \label{eq:x(t) parametric}
x \in \XX \triangleq \{ h_\bth \in L_2[0,T_0] : \bth \in \Theta \}
\eeq
where $h_\bth$ is a set of functions parameterized by the vector $\bth$, and $\Theta$ is an open subset of $\RR^K$.

We wish to examine the random process
\beq \label{eq:y=x+w}
y(t) = x(t) + w(t), \quad t \in [0,T_0]
\eeq
where $w(t)$ is continuous-time white Gaussian noise. Recall that formally, it is not possible to define Gaussian white noise over a continuous-time probability space \cite{kuo02}. Instead, we interpret \eqref{eq:y=x+w} as a simplified notation for the equivalent set of measurements
\beq \label{eq:brownian}
z(t) = \int_0^t x(\tau) d\tau + \sigma_c b(t), \quad t \in [0,T_0]
\eeq
where $b(t)$ is a standard Wiener process (also called Brownian motion) \cite{ibragimov81}. It follows that $w(t)$ can be considered as a random process such that, for any $f,g \in L_2$, the inner products $a=\bra f,w\ket$ and $b=\bra g,w\ket$ are zero-mean jointly Gaussian random variables satisfying $\E{ab^*}=\sigma_c^2\bra f,g\ket$ \cite{kuo02}. The subscript $c$ in $\sigma_c$ is meant as a reminder of the fact that $w(t)$ is continuous-time noise. By contrast, when examining samples of the random process $y(t)$, we will also consider digital noise which is added during the sampling process.

In this paper, we consider estimators which are functions either of the entire continuous-time process \eqref{eq:y=x+w} or of some subset of the information present in \eqref{eq:y=x+w}, such as a discrete set of samples of $y(t)$. To treat these two cases in a unified way, let $(\Omega,{\mathscr F})$ be a measurable space and let $\{P_\bth : \bth \in \Theta \}$ be a family of probability measures over $(\Omega,{\mathscr F})$. Let $(\YY,{\mathscr U})$ be a measurable space, and let the random variable $y : \Omega \rightarrow \YY$ denote the measurements. This random variable can represent either $y(t)$ itself or samples of this quantity.

An estimator can be defined in this general setting as a measurable function $\hat{x} : \YY \rightarrow L_2$. The MSE of an estimator $\hat{x}$ at $x$ is defined as
\beq \label{eq:MSEdef}
\MSE(\hat{x},x) \triangleq \E{\|\hat{x}-x\|_{L_2}^2} = \E{\int_0^{T_0} |\hat{x}(t)-x(t)|^2 dt}.
\eeq
An estimator $\hat{x}$ is said to be unbiased if
\beq \label{eq:def unbiased}
\E{\hat{x}(t)} = x(t) \text{ for all $x \in \XX$ and almost all $t \in [0,T_0]$.}
\eeq

In the next section, we demonstrate the applicability of our model by reviewing several scenarios which can be formulated using the FRI framework. Some of these settings will also be used in the sequel to exemplify our theoretical results.

\section{Types of FRI Signals}
\label{se:types}

Numerous FRI signal structures have been proposed and analyzed in the sampling literature. Whereas most of these can be treated within our framework, some FRI structures do not conform exactly to our problem setting. Thus, before delving into the derivation of the CRB, we first provide examples for scenarios that can be analyzed via our model and discuss some of its limitations.

\subsection{Shift-Invariant Spaces}
\label{ss:sis}
Consider the class of signals that can be expressed as
\begin{equation}\label{eq:SI}
x(t) = \sum_{m\in\ZZ}a[m]g(t-mT)
\end{equation}
with some arbitrary square-integrable sequence $\{a[m]\}_{m \in \ZZ}$, where $g(t)$ is a given pulse in $L_2(\RR)$ and $T>0$ is a given scalar. This set of signals is a linear subspace of $L_2(\RR)$, which is often termed a \emph{shift-invariant} (SI) space \cite{unser00, eldar09e}. The class of functions that can be represented in the form \eqref{eq:SI} is quite large. For example, choosing $g(t)=\sinc(t/T)$ leads to the subspace of $\pi/T$-bandlimited signals. Other important examples include the space of spline functions (obtained by letting $g(t)$ be a B-spline function) and communication signals such as pulse-amplitude modulation (PAM) and quadrature amplitude modulation (QAM). Reconstruction in SI spaces from noiseless samples has been addressed in \cite{Unser94,eldar06b} and extended to the noisy setting in \cite{unser05,eldar06c,ramani08}.

Intuitively, every signal lying in a SI space with spacing $T$ has one degree of freedom per $T$ seconds (corresponding to one coefficient from the sequence $\{a[m]\}$). It is thus tempting to regard the rate of innovation of such signals as $1/T$. However, this is only true in an asymptotic sense and for compactly supported pulses $g(t)$. For any finite window size $T_0$, the $T_0$-local rate of innovation $\rho_{T_0}$ is generally larger. Specifically, suppose that the support of $g(t)$ is contained in $[t_a,t_b]$ and consider intervals of the form $[t,t+MT]$, where $M$ is an integer. Then, due to the overlaps of the pulses, for any such interval we can only assure that there are no more than $M+\lceil(t_b-t_a)/T\rceil$ coefficients affecting the values of $x(t)$. Thus, the $MT$-local rate of innovation of signals of the form \eqref{eq:SI} is given by
\begin{equation}\label{eq:rhoSI}
\rho_{MT}=\frac{1}{T}\left(1+\frac{\lceil\frac{t_b-t_a}{T}\rceil}{M}\right).
\end{equation}
In particular, signals of the form \eqref{eq:SI} having a generator $g(t)$ which is not compactly supported have an infinite $T_0$-local rate of innovation, for any finite $T_0$. This is the case, for example, with bandlimited signals, which are therefore not FRI functions under our definition. As will be discussed in the sequel, this is not a flaw of the definition we use for the rate of innovation. Rather, it reflects the fact that it is impossible to recover any finite-duration segment $[T_1,T_2]$ of such signals from a finite number of measurements.

\subsection{Nonlinearly-Distorted Shift-Invariant Spaces}
In certain communication scenarios, nonlinearities are introduced in order to avoid amplitude clipping, an operation known as companding \cite{Landau61}. When the original signal lies in a SI space, the resulting transmission takes the form
\begin{equation}\label{eq:nlSI}
x(t) = r\!\left(\sum_{m\in\ZZ}a[m]g(t-mT)\right),
\end{equation}
where $r(\cdot)$ is a nonlinear, invertible function. Clearly, the $MT$-local rate of innovation $\rho_{MT}$ of this type of signals is the same as that of the underlying SI function, and is thus given by \eqref{eq:rhoSI}. The recovery of nonlinearly distorted SI signals from noiseless samples was treated in \cite{Landau61,Sandberg63,Dvorkind08,Faktor10}. We are not aware of research works treating the noisy case.

\subsection{Union of Subspaces}
\label{ss:union}
Much of the FRI literature treats signal classes which are unions of subspaces \cite{eldar09, gedalyahu09, gedalyahu11, mishali09}. We now give examples of a few of these models.

\subsubsection{Finite Union of Subspaces}
There are various situations in which a continuous-time signal is known to belong to one of a finite set of spaces. One such signal model is described by
\begin{equation}
x(t) = \sum_{m\in\ZZ}\sum_{k=1}^K a_k[m] g_k(t-mT),
\end{equation}
where $\{g_k(t)\}_{k=1}^K$ are a set of generators. In this model it is assumed that only $L<K$ out of the $K$ sequences $\{a_1[m]\}_{m\in\ZZ},\ldots,\{a_K[m]\}_{m\in\ZZ}$ are not identically zero \cite{eldar09b}. Therefore, the signal $x(t)$ is known to reside in one of $\binom{K}{L}$ spaces, each of which is spanned by an $L$-element subset of the set of generators $\{g_k(t)\}_{k=1}^K$. This class of functions can be used to describe multiband signals \cite{mishali09,mishali09b}. However, the discrete nature of these models precludes analysis using the differential tools employed in the remainder of this paper. Therefore, in this work we will focus on infinite unions of subspaces.

\subsubsection{Single-Burst Channel Sounding}
In certain medium identification and channel sounding scenarios, the echoes of a transmitted pulse $g(t)$ are analyzed to identify the positions and reflectance coefficients of scatterers in the medium \cite{proakis95, tur11}. In these cases, the received signal has the form
\begin{equation} \label{eq:nonper}
x(t) = \sum_{\ell=1}^L a_\ell g(t-t_\ell),
\end{equation}
where $L$ is the number of scatterers and the amplitudes $\{a_\ell\}_{\ell=1}^L$ and time-delays $\{t_\ell\}_{\ell=1}^L$ correspond to the reflectance and location of the scatterers. Such signals can be thought of as belonging to a union of subspaces, where the parameters $\{t_\ell\}_{\ell=1}^L$ determine an $L$-dimensional subspace, and the coefficients $\{a_\ell\}_{\ell=1}^L$ describe the position within the subspace. In contrast with the previous example, however, in this setting we have a union of an infinite number of subspaces, since there are infinitely many possible values for the parameters $t_1, \ldots, t_L$.

In this case, for any window of size $T_0>\max_\ell\{t_\ell\}-\min_\ell\{t_\ell\}$, the $T_0$-local rate of innovation is given by
\begin{equation}\label{eq:rhoSingleBurst}
\rho_{T_0} = \frac{2L}{T_0}.
\end{equation}

\subsubsection{Periodic Channel Sounding}
Occasionally, channel sounding techniques consist of repeatedly probing the medium \cite{bruckstein85}. Assuming the medium does not change throughout the experiment, the result is a periodic signal
\begin{equation} \label{eq:per}
x(t) = \sum_{m\in\ZZ}\sum_{\ell=1}^L a_\ell g(t-t_\ell-mT).
\end{equation}
As before, the set $\XX$ of feasible signals is an infinite union of finite-dimensional subspaces in which $\{t_\ell\}_{\ell=1}^L$ determine the subspace and $\{ a_\ell \}_{\ell=1}^L$ define the position within the subspace. The $T_0$-local rate of innovation in this case coincides with \eqref{eq:rhoSingleBurst}.

\subsubsection{Semi-Periodic Channel Sounding}
There are situations in which a channel consists of $L$ paths whose amplitudes change rapidly, but the time delays can be assumed constant throughout the duration of the experiment \cite{bruckstein85, gedalyahu09, bajwa10}. In these cases, the output of a channel sounding experiment will have the form
\begin{equation} \label{eq:semiper}
x(t) = \sum_{m\in\ZZ}\sum_{\ell=1}^L a_\ell[m] g(t-t_\ell-mT),
\end{equation}
where $a_\ell[m]$ is the amplitude of the $\ell$th path at the $m$th probing experiment. This is, once again, a union of subspaces, but here each subspace is infinite-dimensional, as it is determined by the infinite set of parameters $\{a_\ell[m]\}$. In this case, the $MT$-local rate of innovation can be shown to be
\begin{equation}\label{eq:rhoKfir}
\rho_{MT}=\frac{L}{T}\left(1+\frac{1+\lceil\frac{t_b-t_a}{T}\rceil}{M}\right).
\end{equation}

\subsubsection{Multiband Signals}
Multiuser communication channels are often characterized by a small number of utilized subbands interspersed by large unused frequency bands \cite{mishali09b}. The resulting signal can be described as
\begin{equation} \label{eq:multiband}
x(t) = \sum_{n\in\ZZ}\sum_{\ell=1}^L a_\ell[n] g(t-nT) e^{j\omega_\ell t},
\end{equation}
where $\{a_\ell[n]\}_{n\in\ZZ}$ is the data transmitted by the $\ell$th user, and $\omega_\ell$ is the corresponding carrier frequency. In some cases the transmission frequencies are unknown \cite{mishali09,mishali09b}, resulting in an infinite union of infinite-dimensional subspaces. This setting is analogous in many respects to the semi-periodic channel sounding case; in particular, the $MT$-local rate of innovation can be shown to be the same as that given by \eqref{eq:rhoKfir}.

\section{Summary of Main Results}
\label{se:summary}

Before delving into the mathematical details, we provide in this section a high-level overview of our main contributions and summarize the resulting conclusions.

The overarching objective of this paper is to design and analyze sampling schemes for reconstructing FRI signals from noisy measurements. This goal is accomplished in three stages. First, we identify the best achievable MSE for estimating an FRI signal $x(t)$ from its continuous-time measurements $y(t) = x(t) + w(t)$, providing a fundamental lower bound which is independent of the sampling method. We then compare this continuous-time bound with the lowest possible MSE for a given sampling scheme, thus measuring the loss entailed in any particular technique. Finally, we provide a mechanism for choosing the optimal sampling kernels (in a specific Bayesian sense) utilizing a pre-specified sampling rate budget. Our results can be applied to specific families of FRI signals, but they also yield some general conclusions as to the relative difficulty of various classes of estimation problems. These general observations are summarized below.

\subsection{Continuous-Time Bound}

Our first goal in this paper is to derive the continuous-time CRB, which defines a fundamental limit on the accuracy with which an FRI signal can be estimated, regardless of the sampling technique. This bound turns out to have a particularly simple closed form expression which depends on the rate of innovation, but not on the class $\XX$ of FRI signals being estimated. Specifically, under suitable regularity conditions, the MSE of any unbiased estimator $\hat{x}$ satisfies
\beq
\frac{1}{T_0} \MSE(\hat{x},x) \ge \rho_{T_0} \sigma_c^2.
\eeq
Thus, the rate of innovation can be given a new interpretation as the ratio between the best achievable MSE and the noise variance $\sigma_c^2$. This is to be contrasted with the characterization of the rate of innovation in the noise-free case as the lowest sampling rate allowing for perfect recovery of the signal; indeed, when noise is present, perfect recovery is no longer possible.

\subsection{Bound for Sampled Measurements}

We next consider lower bounds for estimating $x(t)$ from \emph{samples} of the signal $y(t)$. In this setting, the samples inherit the noise $w(t)$ embedded in the signal $y(t)$, and may suffer from additional discrete-time noise, for example, due to quantization. We derive the CRB for estimating $x(t)$ from sampled measurements in the presence of both types of noise. However, the combination of the two noise models complicates the mathematical analysis. Consequently, since the sampling noise model has been previously analyzed \cite{stoica89, berent10}, we focus in this paper on the assumption that the discrete-time noise is negligible.

In this setting, the sampled CRB can be designed so as to converge to the continuous-time bound as the sampling rate increases. Moreover, if the family $\XX$ of FRI signals is contained in a finite-dimensional subspace $\MM$ of $L_2$, then a sampling scheme achieving the continuous-time CRB can be constructed. Such a sampling scheme is obtained by choosing kernels which span the subspace $\MM$, and yields samples which fully capture the information present in the signal $y(t)$. Contrariwise, if $\XX$ is not contained in a finite-dimensional subspace, then no finite-rate sampling method achieves the continuous-time CRB. In this case, any increase in the sampling rate can improve performance, and the continuous-time bound is obtained only asymptotically.

It is interesting to examine this distinction from a union of subspaces viewpoint. Suppose that, as in the examples of Section~\ref{ss:union}, the family $\XX$ can be described as a union of an infinite number of subspaces $\{ \UU_\alpha \}$ indexed by the continuous parameter $\alpha$, so that
\beq
\XX = \bigcup_\alpha \UU_\alpha.
\eeq
In this case, a finite sampling rate captures all of the information present in the signal if and only if
\beq \label{eq:sum subsp}
\operatorname{dim} \left( \sum_\alpha \UU_\alpha \right) < \infty
\eeq
where $\operatorname{dim}(\MM)$ is the dimension of the subspace $\MM$. By contrast, in the noise-free case, it has been previously shown \cite{lu08b} that the number of samples required to recover $x(t)$ is given by
\beq \label{eq:max dim}
\max_{\alpha_1, \alpha_2} \operatorname{dim}(\UU_{\alpha_1} + \UU_{\alpha_2}),
\eeq
i.e., the largest dimension among sums of \emph{two} subspaces belonging to the union. In general, the dimension of \eqref{eq:sum subsp} will be much higher than \eqref{eq:max dim}, illustrating the qualitative difference between the noisy and noise-free settings. For example, if the subspaces $\UU_\alpha$ are finite-dimensional, then \eqref{eq:max dim} is also necessarily finite, whereas \eqref{eq:sum subsp} need not be.

Nevertheless, one may hope that the structure embodied in $\XX$ will allow \emph{nearly} optimal recovery using a sampling rate close to the rate of innovation. This is certainly the case in many noise-free FRI settings. For example, there exist techniques which recover the pulse stream \eqref{eq:nonper} from samples taken at the rate of innovation, despite the fact that in this case $\XX$ is typically not contained in a finite-dimensional subspace. However, this situation often changes when noise is added, in which case standard techniques improve considerably under oversampling. This empirical observation can be quantified using the CRB: as we show, the CRB for samples taken at the rate of innovation is substantially higher in this case than the optimal, continuous-time bound. This demonstrates that the sensitivity to noise is a fundamental aspect of estimating signals of the form \eqref{eq:nonper}, rather than a limitation of existing algorithms. On the other hand, other FRI models, such as the semi-periodic pulse stream \eqref{eq:semiper}, exhibit considerable noise resilience, and indeed in these cases the CRB converges to the continuous-time value much more quickly.

As we discuss in Section~\ref{ss:nonper semiper}, the different levels of robustness to noise can be explained when the signal models are examined in a union of subspaces context. In this case, the parameters $\bth$ defining $x(t)$ can be partitioned into parameters defining the subspace $\UU_\alpha$ and parameters pinpointing the position within the subspace. Our analysis hints that estimation of the position within a subspace is often easier than estimation of the subspace itself. Thus, when most parameters are used to select an intra-subspace position, estimation at the rate of innovation is successful, as occurs in the semi-periodic case \eqref{eq:semiper}. By contrast, when a large portion of the parameters define the subspace in use, a sampling rate higher than the rate of innovation is necessary; this is the case in the non-periodic pulse stream \eqref{eq:nonper}, wherein $\bth$ is evenly divided among subspace-selecting parameters $\{t_\ell\}$ and intra-subspace parameters $\{a_\ell\}$. Thus we see that the CRB, together with the union of subspaces viewpoint, provide valuable insights into the relative degrees of success of various FRI estimation techniques.

\subsection{Choosing the Sampling Kernels}

In some cases, one may choose the sampling system according to design specifications, leading to the question: What sampling kernels should be chosen given an allotted number of samples? We tackle this problem by adopting a Bayesian framework, wherein the signal $x(t)$ is a random process distributed according to a known prior distribution. We further assume that both the sampling and reconstruction techniques are linear. While nonlinear reconstruction methods are often used for estimating FRI signals, this assumption is required for analytical tractability, and is used only for the purpose of identifying sampling kernels. Once these kernels are chosen, they can be used in conjunction with nonlinear reconstruction algorithms.

Under these assumptions, we identify the sampling kernels yielding the minimal MSE\@. An additional advantage of our assumption of linearity is that in this case, the optimal kernels depend only on the autocorrelation
\beq
R_X(t,\tau) = \E{x(t) x^*(\tau)}
\eeq
of the signal $x(t)$, rather than on higher-order statistics. Indeed, given a budget of $N$ samples, the optimal sampling kernels are given by the $N$ eigenfunctions of $R_X$ corresponding to the $N$ largest eigenvalues. This is reminiscent of the Karhunen--Lo\`eve transform (KLT), which can be used to identify the optimal sampling kernels in the noiseless setting. However, in our case, shrinkage is applied to the measurements prior to reconstruction, as is typically the case with Bayesian estimation of signals in additive noise.

A setting of particular interest occurs when the autocorrelation $R_X$ is cyclic, in the sense that
\beq
R_X(t,\tau) = R_X((t-\tau) \operatorname{mod} T)
\eeq
for some $T$. This scenario occurs, for example, in the periodic pulse stream \eqref{eq:per} and the semi-periodic pulse stream \eqref{eq:semiper}, assuming a reasonable prior distribution on the parameters $\bth$. It is not difficult to show that the eigenfunctions of $R_X$ are given, in this case, by the complex exponentials
\beq \label{eq:expon}
\psi_n(t) = \frac{1}{\sqrt{T}} e^{j \frac{2\pi}{T} n t}, \quad n \in \ZZ.
\eeq
Furthermore, in the case of the periodic and semi-periodic pulse streams, the magnitudes of the eigenvalues of $R_X$ are directly proportional to the magnitudes of the respective Fourier coefficients of the pulse shape $g(t)$. It follows that the optimal sampling kernels are the exponentials \eqref{eq:expon} corresponding to the largest Fourier coefficients of $g(t)$. This result is encouraging in light of recently proposed FRI reconstruction techniques which utilize exponential sampling kernels \cite{gedalyahu11}, and demonstrates the suitability of the Bayesian approach for designing practical estimation kernels.

\section{Mathematical Prerequisites: CRB for General Parameter Spaces}
\label{se:math}

In statistics and signal processing textbooks, the CRB is typically derived for parameters belonging to a finite-dimensional Euclidean space \cite{kay93, lehmann98, shao03}. However, this result is insufficient when it is required to estimate a parameter $x$ belonging to other Hilbert spaces, such as the $L_2$ space defined above. When no knowledge about the structure of the parameter $x$ is available, a bound for estimating $x(t)$ from measurements contaminated by \emph{colored} noise was derived in \cite{vantrees68}. However, this bound does not hold when the noise $w(t)$ is white. Indeed, in the white noise case, it can be shown that no finite-MSE unbiased estimators exist, unless further information about $x(t)$ is available. For example, the naive estimator $\hat{x}(t) = y(t)$ has an error $\hat{x}(t) - x(t)$ equal to $w(t)$, whose variance is infinite.

In our setting, we are given the additional information that $x$ belongs to the constraint set $\XX$ of \eqref{eq:x(t) parametric}. To the best of our knowledge, the CRB has not been previously defined for any type of constraint set $\XX \in L_2$, a task which will be accomplished in the present section. As we show below, a finite-valued CRB can be constructed by requiring unbiasedness only within the constraint set $\XX$, as per \eqref{eq:def unbiased}. As we will see, the CRB increases linearly with the dimension of the manifold $\XX$. Thus, in particular, the CRB is infinite when $\XX = L_2$. However, for FRI signals, the dimension of $\XX$ is finite by definition, implying that a finite-valued CRB can be constructed. Although the development of this bound invokes some deep concepts from measure theory, it is a direct analog of the CRB for finite-dimensional parameters \cite[Th.~2.5.15]{lehmann98}.

To derive the bound in the broadest setting possible, in this section we temporarily generalize the scenario of Section~\ref{se:setting}, and consider estimation of a parameter $x$ belonging to an arbitrary measurable and separable Hilbert space $\HH$. The MSE of an estimator $\hat{x}$ in this setting is defined as $\MSE(\hat{x},x) = \E{\|\hat{x}-x\|_\HH^2}$. The concept of bias can similarly be extended if one defines the expectation $\E{v}$ of a random variable $v : \Omega \rightarrow \HH$ as an element $k \in \HH$ such that $\bra k, \varphi \ket = \E{ \bra v, \varphi \ket }$ for any $\varphi \in \HH$. If no such element exist, the expectation is said to be undefined.

The derivation of the CRB requires the existence of a ``probability density'' $p_\bth(y)$ (more precisely, a Radon--Nikodym derivative) which is differentiable with respect to $\bth$, and such that its differentiation with respect to $\bth$ can be interchanged with integration with respect to $y$. The CRB also requires the mapping $h_\bth$ between $\bth$ and $x$ to be non-redundant and differentiable. The formal statement of these regularity conditions is specified below. For the measurement setting \eqref{eq:y=x+w}, with reasonable mappings $h_\bth$, these conditions are guaranteed to hold, as we will demonstrate in the sequel; in this section, however, we list these conditions in full so that a more general statement of the CRB will be possible.

\begin{enumerate}[P1)]
\item \label{ass:dominating measure}
There exists a value $\bth_0 \in \Theta$ such that the measure $P_{\bth_0}$ dominates $\{ P_\bth : \bth \in \Theta \}$. In other words, there exists a Radon--Nikodym derivative $p_\bth(y) \triangleq dP_\bth / dP_{\bth_0}$ such that, for any event $A \in {\mathscr U}$,
\beq
P_\bth(A) = \int_A p_\bth(y) \Pdy.
\eeq

\item \label{ass:cont diff}
For all $y$ such that $p_\bth(y)>0$, the functions $p_\bth(y)$ and $\log p_\bth(y)$ are continuously differentiable with respect to $\bth$. We denote by $\partial p_\bth(y) / \partial \bth$ and $\partial \log p_\bth(y) / \partial \bth$ the column vectors of the gradients of these two functions.

\item \label{ass:common support}
The support $\{ y \in \YY : p_\bth(y) > 0 \}$ of $p_\bth(y)$ is independent of $\bth$.

\item \label{ass:dominating func}
There exists a measurable function $q : \YY \times \Theta \rightarrow \RR$ such that for all sufficiently small $\Delta > 0$, for all $i=1,\ldots,K$, for all $y$, and for all $\bth$,
\beq \label{eq:dominated}
\frac 1 \Delta \left| p_{\bth + \Delta \e_i}(y) - p_\bth(y) \right| \le q(y,\bth)
\eeq
and such that for all $\bth$,
\beq \label{eq:q sq int}
\int q^2(y,\bth) \Pdy < \infty.
\eeq
In \eqref{eq:dominated}, $\e_i$ represents the $i$th column of the $K\times K$ identity matrix.

\item \label{ass:fim}
For each $\bth$, the $K \times K$ Fisher information matrix (FIM)
\beq \label{eq:def Jbth}
\J_\bth \triangleq \E{ \left( \pd{\log p_\bth(y)}{\bth} \right) \left( \pd{\log p_\bth(y)}{\bth} \right)^* }
\eeq
is finite and invertible.

\item \label{ass:h diff}
$h_\bth$ is Fr\'echet differentiable with respect to $\bth$, in the sense that for each $\bth$, there exists a continuous linear operator $\partial h_\bth / \partial \bth : \RR^K \rightarrow L_2$ such that, for any sufficiently small $\bdelta \in \RR^K$,
\beq \label{eq:pdht frechet}
\frac{h_{\bth + \bdelta} - h_\bth}{\|\bdelta\|} = \pdht \bdelta + o(\|\bdelta\|) \quad \text{as $\|\bdelta\|\rightarrow 0$}.
\eeq

\item \label{ass:h nullspace}
The null space of the mapping $\partial h_\bth / \partial \bth$ contains only the zero vector. This assumption is required to ensure that the mapping from $\bth$ to $x$ is non-redundant, in the sense that there does not exist a parametrization of $\XX$ in which the number of degrees of freedom is smaller than $K$.
\end{enumerate}

We are now ready to state the CRB for the estimation of a parameter $x \in L_2[0,T]$ parameterized by a finite-dimensional vector $\bth$. The proof of this theorem is given in Appendix~\ref{ap:th:x from bth}.

\begin{theorem} \label{th:x from bth}
Let $\bth \in \Theta$ be a deterministic parameter, where $\Theta$ is an open set in $\RR^K$. Let $\HH$ be a measurable, separable Hilbert space and let $h_\bth$ be a mapping from $\Theta$ to $\HH$. Let $\{ P_\bth : \bth \in \Theta \}$ be a family of probability measures over a measurable space $(\Omega,{\mathscr F})$, and let $y : \Omega \rightarrow \YY$ be a random variable, where $\YY$ is a measurable Hilbert space. Assume regularity conditions P\ref{ass:dominating measure}--P\ref{ass:h diff}. Let $\hat{x} : \YY \rightarrow \HH$ be an unbiased estimator of $x$ from the measurements $y$ such that
\beq \label{eq:finite energy}
\E{\|\hat{x}(y)\|_\HH^2} < \infty.
\eeq
Then, the MSE of $\hat{x}$ satisfies
\beq \label{eq:th:x from bth}
\MSE(\hat{x},x) \ge \Tr\left[ \left( \pdht \right)^* \left( \pdht \right) \J_\bth^{-1} \right]
\eeq
where $\J_\bth$ is the FIM \eqref{eq:def Jbth}.
\end{theorem}

Theorem~\ref{th:x from bth} enables us to obtain a lower bound on the estimation error of $x$ based on the FIM for estimating $\bth$. The latter can often be computed relatively easily since $\bth$ is a finite-dimensional vector. Even more conveniently, the trace on the right-hand side of \eqref{eq:th:x from bth} is taken over a $K \times K$ matrix, despite the involvement of continuous-time operators. Thus, the computation of \eqref{eq:th:x from bth} is often possible either analytically or numerically, a fact which will be used extensively in the sequel.

\section{CRB for Continuous-Time Measurements}
\label{se:cont-time}

We now apply Theorem~\ref{th:x from bth} to the problem of estimating a deterministic signal $x$ from continuous-time measurements $y$ given by \eqref{eq:y=x+w}.

\begin{theorem} \label{th:crb x from y}
Let $x$ be a deterministic function defined by \eqref{eq:x(t) parametric}, where $\bth \in \Theta$ is an unknown deterministic parameter and $\Theta$ is an open subset of $\RR^K$. Suppose that Assumptions P\ref{ass:h diff}--P\ref{ass:h nullspace} are satisfied. Then, the MSE of any unbiased, finite-variance estimator $\hat{x}$ is bounded by
\beq \label{eq:th:crb x from y}
\MSE(\hat{x},x) \ge K \sigma_c^2.
\eeq
\end{theorem}

The bound of Theorem~\ref{th:crb x from y} can be translated to units of the rate of innovation $\rho_{T_0}$ if we assume that the segment $[0,T_0]$ under analysis achieves the maximum \eqref{eq:def rho}, i.e., this is a segment containing the maximum allowed number of degrees of freedom. In this case $\rho_{T_0} = K/T_0$, and any unbiased estimator $\hat{x}(t)$ satisfies
\beq \label{eq:rho sigma^2}
\frac{\frac{1}{T_0} \E{\int_{0}^{T_0} |x(t)-\hat{x}(t)|^2 dt}}{\sigma_c^2} \ge \rho_{T_0}.
\eeq
In the noisy setting, $\rho_{T_0}$ loses its meaning as a lower bound on the sampling rate required for perfect recovery, since the latter is no longer possible at any sampling rate. On the other hand, it follows from \eqref{eq:rho sigma^2} that the rate of innovation gains an alternative meaning; namely, $\rho_{T_0}$ is a lower bound on the ratio between the average MSE achievable by any unbiased estimator and the noise variance $\sigma_c^2$, \emph{regardless of the sampling method}.

Before formally proving Theorem~\ref{th:crb x from y}, note that \eqref{eq:th:crb x from y} has an intuitive geometric interpretation. Specifically, the constraint set $\XX$ is a $K$-dimensional differential manifold in $L_2[0,T]$. In other words, for any point $x \in \XX$, there exists a $K$-dimensional subspace $\UU$ tangent to $\XX$ at $x$. We refer to $\UU$ as the feasible direction subspace \cite{ben-haim09b}: any perturbation of $x$ which remains within the constraint set $\XX$ must be in one of the directions in $\UU$. Formally, $\UU$ can be defined as the range space of $\partial h_\bth / \partial \bth$.

If one wishes to use the measurements $y$ to distinguish between $x$ and its local neighborhood, then it suffices to observe the projection of $y$ onto $\UU$. Projecting the measurements onto $\UU$ removes most of the noise, retaining only $K$ independent Gaussian components, each having a variance of $\sigma_c^2$. Thus we have obtained an intuitive explanation for the bound of $K \sigma_c^2$ in Theorem~\ref{th:crb x from y}. To formally prove this result, we apply Theorem~\ref{th:x from bth} to the present setting, as follows.

\begin{proof}[Proof of Theorem~\ref{th:crb x from y}]
The problem of estimating the parameters $\bth$ from a continuous-time signal $y(t)$ of the form \eqref{eq:y=x+w} was examined in \cite[Example~I.7.3]{ibragimov81}, where the validity of Assumptions P\ref{ass:dominating measure}--P\ref{ass:dominating func} was demonstrated. It was further shown that the FIM $\J_\bth^{\mathrm{cont}}$ for estimating $\bth$ from $y(t)$ is given by \cite[\emph{ibid.}]{ibragimov81}
\beq \label{eq:Jbth cont}
\J_\bth^{\mathrm{cont}} = \frac{1}{\sigma_c^2} \left(\pdht\right)^* \left(\pdht\right).
\eeq
Our goal will be to use \eqref{eq:Jbth cont} and Theorem~\ref{th:x from bth} to obtain a bound on estimators of the continuous-time function $x(t)$. To this end, observe that the FIM $\J_\bth^{\mathrm{cont}}$ is finite since, by Assumption~P\ref{ass:h diff}, the operator $\partial h_\bth / \partial \bth$ is a bounded operator into $L_2$. Furthermore, by Assumption~P\ref{ass:h nullspace}, $\partial h_\bth / \partial \bth$ has a trivial null space, and thus $\J_\bth^{\mathrm{cont}}$ is invertible. Therefore, Assumption~P\ref{ass:fim} has been demonstrated. We may consequently apply Theorem~\ref{th:x from bth}, which yields
\begin{align}
\MSE&(\hat{x},x) \notag\\
&\ge \sigma_c^2 \Tr \left[ \left( \pdht \right)^* \left( \pdht \right) \left( \left(\pdht\right)^* \left(\pdht\right) \right)^{-1} \right] \notag\\
&=   \sigma_c^2 \Tr \left( \I_{K} \right) \notag\\
&=   K\sigma_c^2
\end{align}
thus completing the proof.
\end{proof}

To illustrate the use of Theorem~\ref{th:crb x from y} in practice, let us consider as a simple example a signal $x(t)$ belonging to a finite-dimensional subspace $\GG$. Specifically, assume that
\beq \label{eq:x_in_G}
x(t) = \sum_{k=1}^K a_k g_k(t)
\eeq
for some coefficient vector $\bth=(a_1,\ldots,a_K)^T$ and a given set of linearly independent functions $\{g_k\}$ spanning $\GG$. This includes, for example, families of shift-invariant subspaces with a compactly supported generator (see Section~\ref{ss:sis}). From Theorem~\ref{th:crb x from y}, the MSE of any unbiased estimator of $x$ is bounded by $K \sigma_c^2$, where $K$ is the dimension of the subspace $\GG$. We now demonstrate that this bound is achieved by the unbiased estimator
\beq \label{eq:PG y}
\hat{x} = \PG\, y
\eeq
where $\PG$ is the orthogonal projector onto the subspace $\GG$.

To verify that \eqref{eq:PG y} achieves the CRB, let $G$ denote the set transformation \eqref{eq:set} associated with the functions $\{g_k\}_{k=1}^K$. One may then write $x = G\bth$ and $\PG = G (G^*G)^{-1} G^*$. Thus \eqref{eq:PG y} becomes
\begin{align}
\hat{x}
&= G (G^*G)^{-1} G^* G\bth + G (G^*G)^{-1} G^* w \notag\\
&= G\bth + \PG\, w
\end{align}
and therefore
\beq
\E{\|\hat{x} - x\|_{L_2}^2} = \E{\|\PG\, w\|_{L_2}^2}.
\eeq
Since $\GG$ is a $K$-dimensional subspace, it is spanned by a set of $K$ orthonormal\footnote{We require the new functions $u_1,\ldots,u_K$ since the functions $g_1, \ldots, g_K$ are not necessarily orthonormal. The choice of non-orthonormal functions $g_1, \ldots, g_K$ will prove useful in the sequel.} functions $u_1, \ldots, u_K \in L_2$. Thus
\beq
\E{\|\PG\, w\|_{L_2}^2} = \sum_{k=1}^K \E{\left| \bra w,u_k \ket \right|^2} = K \sigma_c^2
\eeq
which demonstrates that $\hat{x}$ indeed achieves the CRB in this case.

In practice, a signal is not usually estimated directly from its continuous-time measurements. Rather, the signal $y(t)$ is typically sampled and digitally manipulated. In the next section, we will compare the results of Theorem~\ref{th:crb x from y} with the performance achievable from sampled measurements, and demonstrate that in some cases, a finite-rate sampling scheme is sufficient to achieve the continuous-time bound of Theorem~\ref{th:crb x from y}.

\section{CRB for Sampled Measurements}
\label{se:sampled}

In this section, we consider the problem of estimating $x(t)$ of \eqref{eq:x(t) parametric} from a finite number of samples of the process $y(t)$ given by \eqref{eq:y=x+w}. Specifically, suppose our measurements are given by
\beq \label{eq:c_n}
c_n = \bra y, s_n \ket + v_n = \int_{0}^{T_0} y(t) s^*_n(t) dt + v_n, \quad n=1,\ldots,N
\eeq
where $\{s_n\}_{n=1}^N \subset L_2[0,T_0]$ are sampling kernels, and $v_n$ is a discrete white Gaussian noise process, independent of $w(t)$, having mean zero and variance $\sigma_d^2$. Note that the model \eqref{eq:c_n} includes both continuous-time noise, which is present in the signal $y(t) = x(t) + w(t)$ prior to sampling, and digital noise $v_n$, which arises from the sampling process, e.g., as a result of quantization. In this section, we will separately examine the effect of each of these noise components.

From \eqref{eq:y=x+w} and \eqref{eq:c_n}, it can be seen that the measurements $c_1,\ldots,c_N$ are jointly Gaussian with mean
\beq \label{eq:mu_n}
\mu_n \triangleq \E{c_n} = \bra x, s_n \ket
\eeq
and covariance
\beq \label{eq:gamma}
\Gamma_{ij} \triangleq \Cov(c_i,c_j) = \sigma_c^2 \bra s_i, s_j \ket + \sigma_d^2 \delta_{ij}.
\eeq

A somewhat unusual aspect of this estimation setting is that the choice of the sampling kernels $s_n(t)$ affects not only the measurements obtained, but also the statistics of the noise. One example of the impact of this fact is the following. Suppose first that no digital noise is present, i.e., $\sigma_d=0$, and consider a modified set of sampling kernels $\{ \tilde{s}_n(t) \}_{n=1}^N$ which are an invertible linear transformation of $\{ s_n(t) \}_{n=1}^N$, so that
\beq
\tilde{s}_n(t) = \sum_{i=1}^N \B_{ni} s_i(t)
\eeq
where $\B \in \RR^{N \times N}$ is an invertible matrix. Then, the resulting measurements $\tilde{\c}$ are given by $\tilde{\c} = \B \c$, and similarly the original measurements $\c$ can be recovered from $\tilde{\c}$. It follows that these settings are equivalent in terms of the accuracy with which $x$ can be estimated. In particular, the FIM for estimating $x$ in the two settings is identical \cite[Th.~I.7.2]{ibragimov81}.

When digital noise is present in addition to continuous-time noise, the sampling schemes $\{s_n(t)\}$ and $\{\tilde{s}_n(t)\}$ are no longer necessarily equivalent, since the gain introduced by the transformation $\B$ will alter the ratio between the energy of the signal and the digital noise. The two estimation problems are then equivalent if and only if $\B$ is a unitary transformation.

How should one choose the space $\SS = \spn\{s_1,\ldots,s_N\}$ spanned by the sampling kernels? Suppose for a moment that there exist elements in the range space of $\partial h_\bth/\partial \bth$ which are orthogonal to $\SS$. This implies that one can perturb $x$ in such a way that the constraint set $\XX$ is not violated, without changing the distribution of the measurements $\c$. This situation occurs, for example, when the number of measurements $N$ is smaller than the dimension $K$ of the parametrization of $\XX$. While it may still be possible to reconstruct some of the information concerning $x$ from these measurements, this is an undesirable situation from an estimation point of view. Thus we will assume henceforth that
\beq \label{eq:ident cond}
\Ra{\pdht} \cap \SS^\perp = \{ \zero \}.
\eeq

As an example of the necessity of the condition \eqref{eq:ident cond}, consider again the signal \eqref{eq:x_in_G}, which belongs to a $K$-dimensional subspace $\GG \subset L_2$ spanned by the functions $g_1,\ldots,g_K$. In this case it is readily seen that for any vector $\bv$
\beq \label{eq:pdht GG}
\pdht \bv = \sum_{k=1}^K \bv_k g_k(t).
\eeq
Since the functions $\{g_k\}$ span the space $\GG$, this implies that $\Ra{\partial h_\bth / \partial \bth} = \GG$, and therefore the condition \eqref{eq:ident cond} can be written as
\beq \label{eq:G cap Sperp}
\GG \cap \SS^\perp = \{ \zero \}
\eeq
which is a standard requirement in the design of a sampling system for signals belonging to a subspace $\GG$ \cite{eldar03}.

By virtue of Theorem~\ref{th:x from bth}, a lower bound on unbiased estimation of $x$ can be obtained by first computing the FIM $\J_\bth^{\mathrm{samp}}$ for estimating $\bth$ from $\c$. This yields the following result. For simplicity of notation, in this theorem we assume that the function $h_\bth$ and the sampling kernels $s_n$ are real. If complex sampling kernels are desired (as will be required in the sequel), the result below can still be used by translating each measurement to an equivalent pair of real-valued samples.

\begin{theorem}
\label{th:crb sampled}
Let $x$ be a deterministic real function defined by \eqref{eq:x(t) parametric}, where $\bth \in \Theta$ is an unknown deterministic parameter and $\Theta$ is an open subset of $\RR^K$. Assume regularity conditions P\ref{ass:h diff}--P\ref{ass:h nullspace}, and let $\hat{x}$ be an unbiased estimator of $x$ from the real measurements $\c = (c_1,\ldots,c_N)^T$ of \eqref{eq:c_n}.
Then, the FIM $\J_\bth^{\mathrm{samp}}$ for estimating $\bth$ from $\c$ is given by
\beq \label{eq:Jbth sampled}
\J_\bth^{\mathrm{samp}} = \left(\pdht\right)^* S \left( \sigma_c^2 S^* S + \sigma_d^2 \I_{N} \right)^{-1} S^* \left(\pdht\right)
\eeq
where $S$ is the set transformation corresponding to the functions $\{s_n\}_{n=1}^N$. If \eqref{eq:ident cond} holds, then $\J_\bth^{\mathrm{samp}}$ is invertible. In this case, any finite-variance, unbiased estimator $\hat{x}$ for estimating $x$ from $\c$ satisfies
\beq \label{eq:th:crb sampled}
\MSE(\hat{x},x) \ge \Tr \left[ \left(\pdht\right)^* \left(\pdht\right) (\J_\bth^{\mathrm{samp}})^{-1} \right].
\eeq
\end{theorem}

\begin{proof}
In the present setting, the FIM $\J_\bth^{\mathrm{samp}}$ is given by \cite{kay93}
\beq \label{eq:Jbth 1}
\J_\bth^{\mathrm{samp}} = \left( \pd{\bmu}{\bth} \right)^* \bGamma^{-1} \left( \pd{\bmu}{\bth} \right)
\eeq
where the matrix $\bGamma \in \RR^{N \times N}$ is defined by \eqref{eq:gamma} and the matrix $\partial \bmu / \partial \bth \in \RR^{N \times K}$ is given by
\beq
\left( \pd{\bmu}{\bth} \right)_{nk} = \pd{\mu_n}{\theta_k}
\eeq
with $\mu_n$ defined in \eqref{eq:mu_n}.

By the definition of the set transformation, the $ij$th element of the $N \times N$ matrix $S^* S$ is given by
\beq
(S^* S)_{ij} = \bra S \e_j, S \e_i \ket = \bra s_j, s_i \ket
\eeq
where $\e_i$ is the $i$th column of the $N \times N$ identity matrix. Therefore, we have
\beq \label{eq:bGamma simp}
\bGamma = \sigma_c^2 S^* S + \sigma_d^2 \I_{N}.
\eeq
Similarly, observe that
\begin{align}
\left( S^* \pdht \right)_{nk} &= \bra \pdht \tilde{\e}_k , S \e_n \ket = \bra \pd{h_\bth}{\theta_k}, s_n \ket
 = \pd{\mu_n}{\theta_k}
\end{align}
where $\tilde{\e}_k$ is the $k$th column of the $K \times K$ identity matrix. Thus
\beq \label{eq:bmu/bth}
\pd{\bmu}{\bth} = S^* \pdht.
\eeq
Substituting \eqref{eq:bGamma simp} and \eqref{eq:bmu/bth} into \eqref{eq:Jbth 1} yields the required expression \eqref{eq:Jbth sampled}.

We next demonstrate that if \eqref{eq:ident cond} holds, then $\J_\bth^{\mathrm{samp}}$ is invertible. To see this, note that from \eqref{eq:Jbth sampled} we have
\beq \label{eq:nullsp 1}
\Nu{\J_\bth^{\mathrm{samp}}} = \Nu{S^* \left(\pdht\right)}.
\eeq
Now, consider an arbitrary function $f \in \Ra{\partial h_\bth / \partial \bth}$. If \eqref{eq:ident cond} holds, then $f$ is not orthogonal to the subspace $\SS$. Therefore, $\bra f, s_n \ket \ne 0$ for at least one value of $n$, and thus by \eqref{eq:set_adjoint}, $S^*f\neq 0$. This implies that
\beq
\Nu{S^* \left(\pdht\right)} = \Nu{\pdht} = \{\zero\}.
\eeq
Combined with \eqref{eq:nullsp 1}, we conclude that $\Nu{\J_\bth^{\mathrm{samp}}} = \{\zero\}$. This demonstrates that $\J_\bth^{\mathrm{samp}}$ is invertible, proving Assumption P\ref{ass:fim}. Moreover, in the present setting, Assumptions P\ref{ass:dominating measure}--P\ref{ass:dominating func} are fulfilled for any value of $\bth_0$ \cite{lehmann98}. Applying Theorem~\ref{th:x from bth} yields \eqref{eq:th:crb sampled} and completes the proof.
\end{proof}

In the following subsections we draw several conclusions from Theorem~\ref{th:crb sampled}.

\subsection{Discrete-Time Noise}
Suppose first that $\sigma_c^2=0$, so that only digital noise is present. This setting has been analyzed previously \cite{stoica89, berent10}, and therefore only briefly examine the contrast with continuous-time noise. When only digital noise is present, its effects can be surmounted either by increasing the gain of the sampling kernels, or by increasing the number of measurements. These intuitive conclusions can be verified from Theorem~\ref{th:crb sampled} as follows. Assume that condition \eqref{eq:ident cond} holds, and consider the modified kernels $\tilde{s}_n(t) = 2s_n(t)$. The set transformation $\tilde{S}$ corresponding to the modified kernels is $\tilde{S} = 2S$, and since $\sigma_c^2=0$, this implies that the FIM obtained from the modified kernels is given by $\tilde{\J}_\bth^{\mathrm{samp}} = 4\J_\bth^{\mathrm{samp}}$. Thus, a sufficient increase in the sampling gain can arbitrarily increase $\J_\bth^{\mathrm{samp}}$ and consequently reduce the bound \eqref{eq:th:crb sampled} arbitrarily close to zero.

Of course, from a practical point of view, increasing the gain also increases the likelihood that the sampled signal will exceed the dynamic range of the quantizer. It is therefore not feasible to arbitrarily increase the sampling gain. As an alternative, it is possible to increase the number of measurements. For example, suppose one simply repeats each measurement twice. Let $S$ and $\tilde{S}$ denote the transformations corresponding to the original and doubled sets of measurements. It can then readily be seen from the definition of the set transformation \eqref{eq:set} and its adjoint \eqref{eq:set_adjoint} that $\tilde{S} \tilde{S}^* = 4 S S^*$. Consequently, by the same argument, in the absence of continuous-time noise one can achieve arbitrarily low error by repeated measurements, or more generally, by increasing the sampling rate.

\subsection{Continuous-Time Noise}
As we have seen, sampling noise can be mitigated by increasing the sampling rate. Furthermore, digital noise is inherently dependent on the sampling scheme being used. Since our goal is to determine the fundamental performance limits regardless of the sampling technique, we will focus here and in subsequent sections on continuous-time noise. Thus, suppose that $\sigma_d^2=0$, so that only continuous-time noise is present. In this case, as we now show, it is generally impossible to achieve arbitrarily low reconstruction error, regardless of the sampling kernels used; indeed, it is never possible to outperform the continuous-time CRB of Section~\ref{se:cont-time}, which is typically nonzero. To see this formally, observe first that in the absence of digital noise, the FIM for estimating $\bth$ can be written as
\begin{align} \label{eq:Jbth samp no dig noise}
\J_\bth^{\mathrm{samp}}
&= \frac{1}{\sigma_c^2} \left(\pdht\right)^* S \left( S^* S \right)^{-1} S^* \left(\pdht\right) \notag\\
&= \frac{1}{\sigma_c^2} \left(\pdht\right)^* \PS \left(\pdht\right)
\end{align}
where $\PS$ is the orthogonal projection onto the subspace $\SS$. It is insightful to compare this expression with the FIM $\J_\bth^{\mathrm{cont}}$ obtained from continuous-time measurements in \eqref{eq:Jbth cont}. In both cases, a lower bound on the MSE for unbiased estimation of $x$ was obtained from $\J_\bth$ by applying Theorem~\ref{th:x from bth}. Consequently, if it happens that $\J_\bth^{\mathrm{cont}} = \J_\bth^{\mathrm{samp}}$, then the continuous-time bound of Theorem~\ref{th:crb x from y} and the sampled bound of Theorem~\ref{th:crb sampled} coincide. Thus, if no digital noise is added, then it is possible (at least in terms of the performance bounds) that estimators based on the samples $\c$ will suffer no degradation compared with the ``ideal'' estimator based on the entire set of continuous-time measurements. This occurs if and only if $\Ra{\partial h_\bth / \partial \bth} \subseteq \SS$; in this case, the projection $\PS$ will have no effect on the FIM $\J_\bth^{\mathrm{samp}}$, which will then coincide with $\J_\bth^{\mathrm{cont}}$ of \eqref{eq:Jbth cont}. In the remainder of this section, we will discuss several cases in which this fortunate circumstance arises.

\subsection{Example: Sampling in a Subspace}
\label{ss:subspace achieve crb}
The simplest situation in which samples provide all of the information present in the continuous-time signal is the case in which $x(t)$ belongs to a $K$-dimensional subspace $\GG$ of $L_2$. This is the case, for example, when the signal lies in a shift-invariant subspace having a compactly supported generator (see Section~\ref{ss:sis}). As we have seen above (cf. \eqref{eq:pdht GG}), in this scenario $\partial h_\bth/\partial \bth$ is a mapping onto the subspace $\GG$. Assuming that there is no discrete-time noise, it follows from \eqref{eq:Jbth samp no dig noise} that the optimal choice of a sampling space $\SS$ is $\GG$ itself. Such a choice requires $N=K$ samples and yields $\J_\bth^{\mathrm{cont}} = \J_\bth^{\mathrm{samp}}$. Of course, such an occurrence is not possible if the sampling process contributes additional noise to the measurements.

In some cases, it may be difficult to implement a set of sampling kernels spanning the subspace $\GG$. It may then be desirable to choose a $K$-dimensional subspace $\SS$ which is close to $\GG$ but does not equal it. We will now compute the CRB for this setting and demonstrate that it can be achieved by a practical estimation technique. This will also demonstrate achievability of the CRB in the special case $\SS=\GG$. We first note from \eqref{eq:set} and \eqref{eq:pdht GG} that $\partial h_\bth / \partial \bth = G$, where $G$ is the set transformation corresponding to the generators $\{g_k\}_{k=1}^K$. Furthermore, it follows from \eqref{eq:G cap Sperp} that $S^* G$ and $G^* S$ are invertible $K \times K$ matrices \cite{eldar03}. Using Theorem~\ref{th:crb sampled}, we thus find that the CRB is given by
\begin{align} \label{eq:bound GS}
\MSE(\hat{x},x)
&\ge \sigma_c^2 \Tr\!\left( G \left( G^* S (S^* S)^{-1} S^* G \right)^{-1} G^* \right) \notag\\
&=   \sigma_c^2 \Tr\!\left( G (S^* G)^{-1} S^* S (G^* S)^{-1} G^* \right).
\end{align}
It is readily seen that when $\SS=\GG$, the bound \eqref{eq:bound GS} reduces to $K \sigma_c^2$, which is (as expected) the continuous-time bound of Theorem~\ref{th:crb x from y}. When $\SS \ne \GG$, the bound \eqref{eq:bound GS} will generally be higher than $K \sigma_c^2$, since $\J_\bth^{\mathrm{samp}}$ of \eqref{eq:Jbth samp no dig noise} will exceed $\J_\bth^{\mathrm{cont}}$ of \eqref{eq:Jbth cont}.
In this case, it is common to use the consistent, unbiased estimator \cite{Unser94, eldar03}
\beq
\hat{x} = G (S^* G)^{-1} \c.
\eeq
As we now show, the bound \eqref{eq:bound GS} is achieved by this estimator. Indeed, observe that $\c = S^* y = S^* G \bth + S^* w$, and thus
\begin{align} \label{eq:GS mse}
\E{\|\hat{x}-x\|_{L_2}^2}
&= \E{\|G (S^* G)^{-1} S^* w\|_{L_2}^2} \notag\\
&= \E{\Tr \! \left( G (S^* G)^{-1} S^* w w^* S (G^* S)^{-1} G^* \right)} \notag\\
&= \Tr \! \left( G (S^* G)^{-1} \Cov(S^* w) (G^* S)^{-1} G^* \right).
\end{align}
Note that $\Cov(S^* w) = \Cov(\c)$, which by \eqref{eq:bGamma simp} is equal to $\sigma_c^2 S^* S$. Substituting this result into \eqref{eq:GS mse} and comparing with \eqref{eq:bound GS} verifies that $\hat{x}$ achieves the CRB\@.

\subsection{Nyquist-Equivalent Sampling}
\label{ss:nyquist equiv}

We refer to situations in which the dimension of the sampling space equals the dimension of the signal space as ``Nyquist-equivalent'' sampling schemes. In the previous section, we saw that Nyquist-equivalent sampling is possible using $K$ samples when the signal lies in a $K$-dimensional subspace $\XX$, and that the resulting system achieves the continuous-time CRB\@. A similar situation occurs when the set of possible signals $\XX$ is a subset of an $M$-dimensional subspace $\MM$ of $L_2$ with $M>K$. In this case, it can be readily shown that $\Ra{\partial h_\bth / \partial \bth} \subseteq \MM$. Thus, by choosing $N=M$ sampling kernels such that $\SS = \MM$, we again achieve $\J_\bth^{\mathrm{cont}} = \J_\bth^{\mathrm{samp}}$, demonstrating that all of the information content in $x$ has been captured by the samples. This is again a Nyquist-equivalent scheme, but the number of samples it requires is higher than the number of parameters $K$ defining the signals. Therefore, in this case it is not possible to sample at the rate of innovation without losing some of the information content of the signal.

In general, the constraint set $\XX$ will not be contained in any finite-dimensional subspace of $L_2$. In such cases, it will generally not be possible to achieve the performance of the continuous-time bound using \emph{any} finite number of samples, even in the absence of digital noise. This implies that in the most general setting, sampling above the rate of innovation can often improve the performance of estimation schemes. This conclusion will be verified by simulation in Section~\ref{se:app}.


\section{Optimal Sampling Kernels: A Bayesian Viewpoint}
\label{se:bayesian}

In this section, we address the problem of designing a sampling method which minimizes the MSE. One route towards this goal could be to minimize the sampled CRB of Theorem~\ref{th:crb sampled} with respect to the sampling space $\SS$. However, the CRB is a function of the unknown parameter vector $\bth$. Consequently, for each value of $\bth$, there may be a different sampling space $\SS$ which minimizes the bound. To obtain a sampling method which is optimal on average over all possible choices of $\bth$, we now make the additional assumption that the parameter vector $\bth$ is random and has a known distribution. Our goal, then, is to determine the sampling space $\SS$ that minimizes the MSE $\E{\|\hat{x}-x\|_{L_2}^2}$ within a class of allowed estimators. Note that the mean is now taken over realizations of both the noise $w(t)$ and the parameter $\bth$.

Since $\bth$ is random, the signal $x(t)$ is random as well. To make our discussion general, we will derive the optimal sampling functions for estimating a general random process $x(t)$ (not necessarily having realizations in $\XX$ of \eqref{eq:x(t) parametric}) from samples of the noisy process $y(t)=x(t)+w(t)$. We will then specialize the results to several specific types of FRI signals and obtain explicit expressions for the optimal sampling kernels in these scenarios.

Let $x(t)$ denote a zero-mean random process defined over $t\in[0,T_0]$, and suppose that its autocorrelation function
\begin{equation}
R_{X}(t,\eta)\triangleq \E{x(t)x^*(\eta)}
\end{equation}
is continuous in $t$ and $\eta$. Our goal is to estimate $x(t)$ based on a finite number $N$ of samples of the signal $y(t)=x(t)+w(t)$, $t\in[0,T_0]$, where $w(t)$ is a white noise process (not necessarily Gaussian) with variance $\sigma_c^2$ which is uncorrelated with $x(t)$. We focus our attention on \emph{linear sampling} schemes, i.e., we assume the samples are given by
\beq \label{eq:c_n no sampling noise}
c_n = \bra y, s_n \ket.
\eeq
Finally, we restrict the discussion to \emph{linear estimation} methods, namely those techniques in which the estimate $\hat{x}(t)$ is constructed as
\begin{equation}\label{eq:hx}
\hat{x}(t) = \sum_{n=1}^N c_n v_n(t),
\end{equation}
for some set of reconstruction functions $\{v_n(t)\}_{n=1}^N$. It is important to note that for any given set of sampling functions $\{s_n(t)\}_{n=1}^N$, the minimum MSE (MMSE) estimator of $x(t)$ is often a nonlinear function of the measurements $\{c_n\}_{n=1}^N$. Indeed, typical FRI reconstruction techniques involve a nonlinear stage. Consequently, restricting the discussion to linear recovery schemes may seem inadequate. However, this choice has two advantages. First, as we will see, the optimal linear scheme is determined only by the second-order statistics of $x(t)$ and $w(t)$, whereas the analysis of nonlinear methods necessitates exact knowledge of their entire distribution functions. Second, it is not the final estimate $\hat{x}(t)$ that interests us in this discussion, but merely the set of optimal sampling functions. Once such a set is determined, albeit from a linear recovery perspective, it can be used in conjunction with existing nonlinear FRI techniques. As we will see in Section~\ref{se:app}, the conclusions obtained through our analysis appear to apply to FRI techniques in general. Under the above assumptions, our goal is to design the sampling kernels $\{s_n(t)\}_{n=1}^N$ and reconstruction functions $\{v_n(t)\}_{n=1}^N$ such that the MSE \eqref{eq:MSEdef} is minimized.

As can be seen from \eqref{eq:c_n no sampling noise}, we assume henceforth that only continuous-time noise is present in the sampling system. The situation is considerably more complicated in the presence of digital noise. First, without digital noise, one must choose only the subspace spanned by the sampling kernels, as the kernels themselves do not affect the performance; this is no longer the case when digital noise is added. Second, digital noise may give rise to a requirement that a particular measurement be repeated in order to average out the noise. This is undesirable in the continuous noise regime, since the repeated measurement will contain the exact same noise realization.

\subsection{Relation to the Karhunen--Lo\`eve Expansion and Finite-Dimensional Generalizations}

The problem posed above is closely related to the Karhunen--Lo\`eve transform (KLT) \cite{loeve48, loeve78}, which is concerned with the reconstruction of a random signal $x(t)$ from its noiseless samples. Specifically, one may express $x(t)$ in terms of a complete orthonormal basis $\{\psi_k(t)\}_{k=1}^\infty$ for $L_2$ as
\begin{equation}\label{eq:KLexpansion}
x(t) = \sum_{k=1}^\infty \bra x,\psi_k \ket  \psi_k(t).
\end{equation}
The goal of the KLT is to choose the functions $\{\psi_k(t)\}_{k=1}^\infty$ such that the MSE resulting from the truncation of this series after $N$ terms is minimal. It is well known that the solution to this problem is given by the $N$-term truncation of the Karhunen--Lo\`eve expansion \cite{loeve48,maitre10}.

Since $R_X(t,\eta)$ is assumed to be continuous in our setting, by Mercer's theorem \cite{maitre10} it possesses a discrete set of eigenfunctions $\{\psi_k(t)\}_{k=1}^\infty$, which constitute an orthonormal basis for $L_2$. These functions satisfy the equations
\begin{align}\label{eq:eigenfunctions}
\lambda_k \psi_k(t) = \int_0^{T_0} R_X(t,\eta) \psi_k(\eta) d\eta,
\end{align}
in which the corresponding eigenvalues $\lambda_1\geq\lambda_2\geq\cdots\geq 0$ are nonnegative and are assumed to be arranged in descending order. With these functions, \eqref{eq:KLexpansion} is known as the Karhunen--Lo\`eve expansion. It can be easily shown that the first $N$ terms in this series constitute the best $N$-term approximation of $x(t)$ in an MSE sense \cite{maitre10}. In other words, in the noiseless case, the optimal sampling and reconstruction functions are $s_n(t) = v_n(t) = \psi_n(t)$.

In our setting, we do not have access to samples of $x(t)$ but rather only to samples of the noisy process $y(t)$. In this case, it is not clear \emph{a priori} whether the optimal sampling and reconstruction filters coincide or whether they match the Karhunen--Lo\`eve expansion of $x(t)$.

The finite-dimensional analogue of our problem, in which $x$, $y$, and $w$ are random vectors taking values in $\CC^M$, was treated in \cite{yamashita96,hua98}. The derivation in these works, however, relied on the low-rank approximation property of the singular-value decomposition (SVD) of a matrix. The generalization of this concept to infinite-dimensional operators is subtle and will thus be avoided here. Instead, we provide a conceptually simple (if slightly cumbersome) derivation of the optimal linear sampling and reconstruction method for noisy signals. As we will see, it still holds that $s_n(t) = \psi_n(t)$, but $v_n(t) = \alpha_n \psi_n(t)$, where $\alpha_n$ is a shrinkage factor depending on the SNR of the $n$th sample.

\subsection{Optimal Sampling in Noisy Settings}

As explained in Section~\ref{se:sampled}, in the absence of discrete-time noise, the MSE is not affected by modifications of the sampling kernels which leave the set $\SS=\spn\{s_1(t),\ldots,s_N(t)\}$ unchanged. Thus, without loss of generality, we constrain $\{s_n(t)\}_{n=1}^N$ to satisfy
\begin{equation}\label{eq:s_const}
\left\langle s_n\,,\,\sigma_c^2 s_m + \int_{0}^{T}R_{X}(\cdot,\tau) s_m(\tau) d\tau \right\rangle = \delta_{m,n}
\end{equation}
for every $m,n=1,\ldots,N$. This can always be done since the operator $R_Y:L_2\rightarrow L_2$ defined by $(R_Yf)(t)=\int_0^T R_X(t,\tau)f(\tau)d\tau+\sigma_c^2f(t)$ is positive definite. This choice is particularly convenient as it results in a set of uncorrelated samples $\{c_n\}$. Indeed
\begin{align}\label{eq:cov_c}
\E{c_m c_n^*}&=\E{\left(\int_{0}^{T} s^*_m(\tau)y(\tau)d\tau\right)\left(\int_{0}^{T} s^*_n(\eta)y(\eta)d\eta\right)^*} \nonumber\\
&= \iint_{0}^{T} s^*_m(\tau) \E{y(\tau)y^*(\eta)} s_n(\eta) d\tau d\eta \nonumber\\
&= \iint_{0}^{T} s^*_m(\tau) R_{X}(\tau,\eta) s_n(\eta) d\tau d\eta \nonumber\\
&\hspace{3cm}+\E{\bra s_n,w \ket \bra s_m,w \ket^*}\nonumber\\
&= \iint_{0}^{T} s^*_m(\tau) R_{X}^*(\eta,\tau) s_n(\eta) d\tau d\eta + \sigma_c^2\left\langle s_n,s_m\right\rangle\nonumber\\
&= \delta_{m,n}.
\end{align}

We are now ready to determine the optimal sampling method. We begin by expressing the MSE \eqref{eq:MSEdef} as
\begin{align}\label{eq:MSE_explicit}
&\int_0^T  \E{\left|x(t)-\hat{x}(t)\right|^2} dt = \int_0^T \E{|x(t)|^2} dt \nonumber\\
&\hspace{1.25cm}  - 2\int_0^T\Re\left\{\E{x^*(t)\hat{x}(t)}\right\}dt  + \int_0^T\E{|\hat{x}(t)|^2}dt.
\end{align}
The first term in this expression does not depend on the choice of $\{s_n(t)\}_{n=1}^N$ and $\{v_n(t)\}_{n=1}^N$, and is therefore irrelevant for our purpose. Substituting \eqref{eq:hx} and \eqref{eq:c_n no sampling noise}, and using the fact that $w(t)$ is uncorrelated with $x(t)$, the second term can be written as
\begin{align}\label{eq:cov_xc}
&\int_0^T2\Re\left\{\E{x^*(t)\sum_{n=1}^N c_n v_n(t)}\right\}dt \nonumber\\
&=\sum_{n=1}^N 2\int_0^T\Re\left\{\E{x^*(t) \int_{0}^{T}\!\! y(\tau)s^*_n(\tau)  d\tau} v_n(t)\right\}dt \nonumber\\
&= \sum_{n=1}^N 2\!\iint_0^T\Re\left\{\E{x^*(t) (x(\tau)+w(\tau)) } s^*_n(\tau) v_n(t)\right\}d\tau dt \nonumber\\
&= \sum_{n=1}^N 2\iint_0^T\Re\left\{s^*_n(\tau) R_{X}(\tau,t) v_n(t)\right\} d\tau dt \nonumber\\
&= \sum_{n=1}^N 2\Re\left\{\left\langle v_n\,,\,\int_0^T R_{X}(\cdot,\tau)s_n(\tau)d\tau\right\rangle\right\}.
\end{align}
Similarly, using the fact that $\{c_n\}_{n=1}^N$ are uncorrelated and have unit variance (see \eqref{eq:cov_c}), the last term in \eqref{eq:MSE_explicit} becomes
\begin{align}\label{eq:last_term}
\int_0^T\E{\left|\sum_{n=1}^N c_n v_n(t)\right|^2}dt &= \sum_{m=1}^N \sum_{n=1}^N \E{c_m^*c_n } \bra v_m, v_n \ket \nonumber\\
&= \sum_{n=1}^N \|v_n\|^2.
\end{align}
Substituting \eqref{eq:cov_xc} and \eqref{eq:last_term} back into \eqref{eq:MSE_explicit}, we conclude that minimization of the MSE is equivalent to minimization of
\begin{align}\label{eq:MSE_simple}
&\sum_{n=1}^N \left(\|v_n\|^2-2\Re\left\{\left\langle v_n,\int_0^T\!\!R_{X}(\cdot,\tau)s_n(\tau)d\tau\right\rangle\right\}\right)
\end{align}
with respect to $\{s_n(t)\}_{n=1}^N$ and $\{v_n(t)\}_{n=1}^N$, subject to the set of constraints \eqref{eq:s_const}.

As a first stage, we minimize \eqref{eq:MSE_simple} with respect to the reconstruction functions $\{v_n(t)\}_{n=1}^N$. To this end, we note that the $n$th summand in \eqref{eq:MSE_simple} is lower bounded by
\begin{align}
&\|v_n\|^2 - 2\|v_n\|\left\|\int_0^T\!\!R_{X}(\cdot,\tau)s_n(\tau)d\tau\right\| \nonumber\\
&\hspace{3.5cm}\geq-\left\|\int_0^T\!\!R_{X}(\cdot,\tau)s_n(\tau)d\tau\right\|^2,
\end{align}
where we used the Cauchy--Schwarz inequality and the fact that $\min_{z}\{z^2-2bz\}=-b^2$. This bound is achieved by choosing
\begin{equation}\label{eq:v_opt}
v_n(t)= \int_0^T R_{X}(t,\tau) s_n(\tau)d\tau,
\end{equation}
thus identifying the optimal reconstruction functions.

Substituting \eqref{eq:v_opt} into \eqref{eq:MSE_simple}, our goal becomes to maximize
\begin{equation}\label{eq:MSE_sub_vn}
\sum_{n=1}^N\left\|\int_0^T R_{X}(\cdot,\tau)s_n(\tau)d\tau\right\|^2
\end{equation}
with respect to the sampling functions $\{s_n(t)\}_{n=1}^N$. As we show in Appendix~\ref{ap:max_MSE_sub_vn}, the maximum of this expression is achieved by any set of kernels of the form
\begin{equation}
s_n(t)=\sum_{k=1}^N \A_{k,n} \left(\lambda_k+\sigma_c^2\right)^{-\frac{1}{2}}\psi_k(t),
\end{equation}
where $\A$ is a unitary $N\times N$ matrix and $\lambda_k$ and $\psi_k(t)$ are the eigenvalues and eigenfunctions of $R_x(t,\eta)$ respectively (see \eqref{eq:eigenfunctions}). In particular, we can choose $\A=\I_N$, leading to
\begin{equation}\label{eq:opt_sn}
s_n(t) = \frac{1}{\sqrt{\lambda_n+\sigma_c^2}}\,\psi_n(t),\quad n=1,\ldots,N.
\end{equation}
From \eqref{eq:v_opt}, the optimal reconstruction kernels are given by
\begin{align}\label{eq:opt_vn}
v_n(t)=\frac{\lambda_n}{\sqrt{\lambda_n+\sigma_c^2}}\,\psi_n(t),\quad n=1,\ldots,N.
\end{align}
The following theorem summarizes the result.

\begin{theorem}\label{th:Bayesian}
Let $x(t)$, $t\in[0,T]$ be a random process whose autocorrelation function $R_X(t,\eta)$ is jointly continuous in $t$ and $\eta$. Assume that $y(t)=x(t)+w(t)$, where $w(t)$ is a white noise process uncorrelated with $x(t)$. Then, among all estimates $\hat{x}(t)$ of $x(t)$ having the form
\beq
\hat{x}(t)=\sum_{n=1}^Nv_n(t)\int_0^Ts^*_n(\tau)y(\tau)dt
\eeq
the MSE \eqref{eq:MSEdef} is minimized with $\{s_n(t)\}_{n=1}^N$ and $\{v_n(t)\}_{n=1}^N$ of \eqref{eq:opt_sn} and \eqref{eq:opt_vn} respectively. In these expressions, $\lambda_n$ and $\psi_n(t)$ are the eigenvalues and eigenfunctions of $R_x(t,\eta)$ respectively (see \eqref{eq:eigenfunctions}).
\end{theorem}

Interestingly, the optimal sampling and reconstruction functions in our noisy setting are similar to those dictated by the KLT\@. The only difference is that in the present scenario, the $n$th sample is shrunk by a factor of $\lambda_n/(\lambda_n+\sigma_c^2)$ prior to reconstruction. This ensures that the low-SNR measurements do not contribute to the recovery as much as their high-SNR counterparts. From the viewpoint of designing the sampling mechanism, however, this difference is of no importance.

As stated above, in practice one would generally favor nonlinear processing of the samples (namely, applying standard nonlinear FRI techniques) rather than a simple element-wise shrinkage. Thus, the importance of Theorem~\ref{th:Bayesian} for our purposes is in identifying that the eigenfunctions of $R_X(t,\tau)$ remain the optimal sampling kernels even in the noisy setting.

\subsection{Example: Sampling in a Subspace}

To demonstrate the utility of Theorem~\ref{th:Bayesian}, we now revisit the situation in which $x(t)$ is given by \eqref{eq:x_in_G} for some set of linearly independent functions $\{g_k(t)\}_{k=1}^K$ spanning a subspace $\GG\in L_2$. We assume that the coefficients $\bth=\{a_1,\ldots,a_K\}^T$ form a zero-mean random vector and denote its autocorrelation matrix by $\R_\bth$. In this case, the signal's autocorrelation function is given by
\begin{align}\label{eq:RxSubspace}
R_X(t,\eta) &= \E{x(t)x^*(\eta)} \nonumber\\
&= \E{\sum_{k=1}^K a_k g_k(t) \sum_{\ell=1}^K a_\ell^* g_k^*(\eta)} \nonumber\\
&= \sum_{k=1}^K\sum_{\ell=1}^K g_k(t) g_\ell^*(\eta) (\R_\bth)_{k,\ell}.
\end{align}
Consequently, the operator $R_X:L_2\rightarrow L_2$ defined by $(R_x h)(t)=\int_0^T R_x(t,\eta)h(\eta)d\eta$ can be expressed as
\begin{equation}\label{eq:GRG}
R_x = G\R_\bth G^*,
\end{equation}
where $G$ is the set transformation \eqref{eq:set} associated with $\{g_k\}_{k=1}^K$.

Now, let $\U$ be a unitary matrix and let $\D$ be a diagonal matrix, such that
\begin{equation}\label{eq:UDU}
\U\D\U^* = (G^*G)^{1/2}\R_\bth(G^*G)^{1/2}.
\end{equation}
Since the dimension of $\Ra{G}$ is $K$, the operator $R_X$ has at most $K$ nonzero eigenvalues $\{\lambda_k\}_{k=1}^K$. Let $\Psi$ denote the set transformation associated with the $N$ eigenfunctions $\{\psi_n\}_{n=1}^N$ corresponding to the $N$ largest eigenvalues, for some $N \le K$. Then, it can be shown that
\begin{equation}\label{eq:psi_n_subspace}
\Psi=G(G^*G)^{-1/2}\U
\end{equation}
and the corresponding eigenvalues are
\begin{equation}
\lambda_n = \D_{n,n}.
\end{equation}
To see this, note that according to \eqref{eq:psi_n_subspace}, $\Psi$ is an isometry, since
\begin{equation}
\Psi^*\Psi=\U^*(G^*G)^{-1/2}G^*G(G^*G)^{-1/2}\U=\U^*\U=\I_K.
\end{equation}
Furthermore, \eqref{eq:GRG} and \eqref{eq:UDU} imply that $R_X=\Psi\D\Psi^*$. Consequently
\begin{align}
R_X \Psi = \Psi\D\Psi^*\Psi = \Psi\D,
\end{align}
which proves the claim.

It is important to emphasize that the $K$ functions $\{\psi_n(t)\}_{n=1}^K$ span $\GG$. Therefore, if one is allowed to take $N=K$ samples, then the optimal choice is a set of kernels that span $\GG$. This conclusion is compatible with the CRB analysis of the previous sections. However, the advantage of the Bayesian viewpoint is that it allows us to identify the optimal sampling space when less than $K$ samples are allowed. For example, suppose that $\{g_n(t)\}$ are orthonormal, and the coefficients $\{a_n\}$ are uncorrelated. Then the optimal sampling space is the one spanned by the $N$ functions $\{g_n(t)\}$ corresponding to the $N$ largest-variance coefficients $\{a_n\}$.

A second example demonstrating the derivation of the optimal sampling kernels will be given in the next section.

\section{Application: Channel Estimation}
\label{se:app}

In this section, we focus on a specific application of FRI signals, namely, that of estimating a signal consisting of a number of pulses having unknown positions and amplitudes \cite{gedalyahu09, gedalyahu11, tur11}. More precisely, we consider periodic signals $x(t)$ of the form \eqref{eq:per}, which were discussed in Section~\ref{ss:union}. These are $T$-periodic pulse sequences, in which the pulse shape $g(t)$ is known, but the amplitudes $\{a_\ell\}$ and delays $\{t_\ell\}$ are unknown. After analyzing periodic signals of this type, we will also compare estimation performance in this case with the semi-periodic family \eqref{eq:semiper}, and attempt to explain the empirically observed differences in stability between these two cases.

By defining the $T$-periodic function $h(t) = \sum_{n \in \ZZ} g(t-nT)$, we can write $x(t)$ of \eqref{eq:per} as
\beq \label{eq:x(t) time delay}
x(t) = \sum_{\ell=1}^L a_\ell h(t - t_\ell).
\eeq
Our goal is now to estimate $x(t)$ from samples of the noisy process $y(t)$ of \eqref{eq:y=x+w}. As before, we will assume that only continuous-time noise is present in the system. Since $x(t)$ is $T$-periodic, it suffices to recover the signal in the region $[0,T]$. In particular, we would like to identify the optimal sampling kernels for this setting, and to compare existing algorithms with the resulting CRB in order to determine when the optimal estimation performance is achieved.

Let
\beq
\tilde{h}_k = \frac 1 T \bra h, \varphi_k \ket , \quad k \in \ZZ
\eeq
be the Fourier series of $h(t)$, where $\varphi_k(t) =  e^{j 2\pi kt/T}$. The Fourier series of $x(t)$ is then given by
\beq \label{eq:x tilde}
\tilde{x}_k \triangleq \frac{1}{T}\bra x, \varphi_k \ket = \tilde{h}_k \sum_{\ell=1}^L a_\ell e^{-j \frac{2\pi}{T} k t_\ell}, \quad k \in \ZZ.
\eeq
Let $\KK = \{ k \in \ZZ : \tilde{h}_k \ne 0 \}$ denote the indices of the nonzero Fourier coefficients of $h(t)$. Suppose for a moment that $\KK$ is finite. It then follows from \eqref{eq:x tilde} that $x(t)$ also has a finite number of nonzero Fourier coefficients. Consequently, the set $\XX$ of possible signals $x(t)$ is contained in the finite-dimensional subspace $\MM = \spn\{ \varphi_k \}_{k \in \KK}$. Therefore, as explained in Section~\ref{ss:nyquist equiv}, choosing the $N = |\KK|$ sampling kernels $\{s_n(t) = e^{-j 2\pi nt/T}\}_{n \in \KK}$ results in a sampled CRB which is equivalent to the continuous-time bound. This result is compatible with recent work demonstrating successful performance of FRI recovery algorithms using exponentials as sampling kernels \cite{gedalyahu11}.

Note, however, that this is a ``Nyquist-equivalent'' sampling scheme, i.e., the number of samples required $N=|\KK|$ is potentially much higher than the number of degrees of freedom $2L$ in the signal $x(t)$ (see Section~\ref{ss:nyquist equiv}). This provides a theoretical explanation of the empirically recognized fact that sampling above the rate of innovation improves the performance of FRI techniques in the presence of noise \cite{tur11}, a fact which stands in contrast to the noise-free performance guarantees of many FRI algorithms.

Moreover, if there exists an infinite number of nonzero coefficients $\tilde{h}_k$, then in general the set $\XX$ will not belong to any finite-dimensional subspace. Consequently, it will not be possible in this case for an algorithm based on a finite number of samples to achieve the performance obtainable from the complete signal $y(t)$. This occurs, for example, whenever the pulse $g(t)$ of \eqref{eq:per} is time-limited. In such cases, any increase in the sampling rate will potentially continue to reduce the CRB, although the sampled CRB will converge to the asymptotic value of $\rho_{T_0} \sigma_c^2$ in the limit as the sampling rate increases.

\subsection{Choosing the Sampling Kernels}

An important question in the current setting is how to choose the sampling kernels so as to achieve the best possible performance under a limited budget of samples. This can be done via the Bayesian analysis provided in Section~\ref{se:bayesian}. Assume, for example, that the time delays $\{t_\ell\}_{\ell=1}^L$ are independently drawn from a uniform distribution over the interval $[0,T]$. Furthermore, suppose that the amplitudes $\{a_\ell\}_{\ell=1}^L$ are mutually uncorrelated zero-mean random variables which are independent of the time delays and have variance $\sigma_a^2$. Then,
\begin{align}
R_X(t,\tau) &= \E{x(t)x^*(\tau)}\nonumber\\
&= \sum_{k=1}^L \sum_{\ell=1}^L \E{a_k a_\ell^*} \E{h(t - t_k)h^*(\tau - t_\ell)}\nonumber\\
&= \sigma_a^2 \sum_{\ell=1}^L \E{h(t - t_\ell)h^*(\tau - t_\ell)}\nonumber\\
&= \sigma_a^2 L \frac{1}{T} \int_0^T h(t - t_\ell)h^*(\tau - t_\ell) dt_\ell \nonumber\\
&= \sigma_a^2 L \sum_{k\in\ZZ} |\tilde{h}_k|^2 e^{j \frac{2\pi}{T} k(t-\tau)},
\end{align}
where we used Parseval's theorem. It is easily verified that the eigenfunctions of $R_X(t,\tau)$ are given by
\begin{equation}
\psi_n(t)=\frac{1}{\sqrt{T}}\,\,e^{j\frac{2\pi}{T}nt},\quad n\in\ZZ
\end{equation}
and the corresponding eigenvalues are
\begin{equation}
\lambda_n=L \sigma_a^2 T |\tilde{h}_n|^2,\quad n\in\ZZ.
\end{equation}
Therefore, the optimal set of $N$ sampling functions is
\beq \label{eq:opt sampling kernels}
s_n(t)=e^{j\frac{2\pi}{T}p_nt}, \quad n=1,\ldots,N
\eeq
where $p_n$ is the index of the $n$th largest Fourier coefficient $|\tilde{h}_{p_n}|$. The optimal linear recovery of $x(t)$ from the resulting samples is given by
\begin{equation}
\hat{x}(t) = \sum_{n=1}^N c_n \frac{L \sigma_a^2 |\tilde{h}_{p_n}|^2}{L \sigma_a^2 T |\tilde{h}_{p_n}|^2+\sigma_c^2}\,e^{j\frac{2\pi}{T}p_nt}.
\end{equation}
The performance of this estimator is poorer than state-of-the-art techniques, due to the restriction to linear reconstruction schemes. We recall that this technique is intended only for selecting the sampling kernels.

The above analysis again lends credence to the recently proposed time-delay estimation technique of Gedalyahu et al.\ \cite{gedalyahu11}, which makes use of complex exponentials as sampling functions. A disadvantage of this algorithm is that it can only handle a set of exponents with successive frequencies, while for general pulses, the indices of the $N$ largest Fourier coefficients may be sporadic. As we will see in Section~\ref{ss:pulse shape}, this limitation may result in deteriorated performance of the algorithm in some cases.

\subsection{Computing the CRB}
\label{ss:compute CRB}

Having identified the optimal sampling kernels \eqref{eq:opt sampling kernels}, we would now like to compute the CRB for estimating $x(t)$ from the resulting samples. In order to compare these results with the continuous-time CRB, we assume that no digital noise is added in the sampling process. However, the calculations described below can be adapted without difficulty to situations containing both continuous-time and digital noise.

We assume for simplicity that $h(t)$ and $\{a_\ell\}$ are real-valued. Nonetheless, the sampling kernels chosen above are complex-valued, implying that Theorem~\ref{th:crb sampled} cannot be directly applied. Yet since $h(t)$ is real-valued, we have $|\tilde{h}_k| = |\tilde{h}_{-k}|$, and consequently the optimal sampling kernels consist of complex conjugate pairs $e^{\pm j 2 \pi nt/T}$. Recall that the sampling kernels can be changed without affecting the CRB, as long as the subspace they span remains constant. Consequently, the CRB can be computed for the equivalent sampling kernels $\sin(2\pi nt/T)$ and $\cos(2\pi nt/T)$, which are real and can therefore be used in conjunction with the results of Section~\ref{se:sampled}. We note that since the transition to these real-valued kernels is unitary, the CRB will not change even if digital noise is added. To be specific without complicating the notation, we assume that $N$ is odd and that the DC component is included among the sampling kernels chosen in \eqref{eq:opt sampling kernels}. We can then define the equivalent set of kernels
\begin{align} \label{eq:fourier sampling}
\tilde{s}_0(t)                 &= 1, \notag\\
\tilde{s}_n(t)                 &= \cos(2\pi p_n t/T),   & n &= 1,\ldots,\frac{N-1}{2}, \notag\\
\tilde{s}_{n+\frac{N+1}{2}}(t) &= \sin(2\pi p_n t/T),   & n &= 1,\ldots,\frac{N-1}{2}.
\end{align}
We further define the parameter vector
\beq
\bth = (a_1, \ldots, a_L, t_1, \ldots, t_L)^T
\eeq
whose length is $K = 2L$.

Theorem~\ref{th:crb sampled} provides a two-step process for computing the CRB of the signal $x(t)$ from its samples. First, the FIM $\J_\bth^\mathrm{samp}$ for estimating $\bth$ is determined. Second, the formula \eqref{eq:th:crb sampled} is applied to compute the CRB\@. While Theorem~\ref{th:crb sampled} also provides a means for calculating $\J_\bth^\mathrm{samp}$, it is more convenient in the present setting to derive the FIM directly. This can be done by calculating the expectations $\mu_n$ of \eqref{eq:mu_n} and applying \eqref{eq:Jbth 1}. In our setting, $\mu_n = \bra x, \tilde{s}_n \ket$ are given by
\begin{align}
\mu_n                 &= \frac{\tilde{x}_{p_n} + \tilde{x}_{-p_n}}{2},     & n &= 0,\ldots,\frac{N-1}{2} \notag\\
\mu_{n+\frac{N+1}{2}} &= \frac{\tilde{x}_{p_n} - \tilde{x}_{-p_n}}{2j},     & n &= 1,\ldots,\frac{N-1}{2}
\end{align}
where $\{\tilde{x}_n\}$ are the Fourier coefficients of $x(t)$. These coefficients depend in turn on the parameter vector $\bth$, as shown in \eqref{eq:x tilde}. Substituting $\mu_n$ into \eqref{eq:Jbth 1} yields a closed-form expression for $\J_\bth^\mathrm{samp}$. Since the resulting formula is cumbersome and not very insightful, it is not explicitly written herein.

To obtain the sampled CRB, our next step is to compute the $2L \times 2L$ matrix
\beq
\M \triangleq \left(\pdht\right)^* \left(\pdht\right).
\eeq
The function $h_\bth : \RR^{2L} \rightarrow L_2$ maps a given parameter vector $\bth$ to the resulting signal $x(t)$ as defined by \eqref{eq:x(t) time delay}. Differentiating this function with respect to $\bth$, we find that the operator $\partial h_\bth / \partial \bth : \RR^{2L} \rightarrow L_2$ is defined by
\begin{align} \label{eq:pdht time delay}
\left(\pdht\right)\bv &= \bv_1 h(t-t_1) + \cdots + \bv_L h(t-t_L) \notag\\
&\hspace{5mm} - \bv_{L+1} a_1 h'(t-t_1) - \cdots - \bv_{2L} a_L h'(t-t_L)
\end{align}
for any vector $\bv \in \RR^{2L}$.

One may now compute the $ik$th element of $\M$ as
\begin{align}
M_{ik} = \e_i^* \left(\pdht\right)^* \left(\pdht\right) \e_k = \bra \pdht \e_i , \pdht \e_k \ket .
\end{align}
Thus, each element of $\M$ is an inner product between two of the terms in \eqref{eq:pdht time delay}. To calculate this inner product numerically for a given function $h(t)$, it is more convenient to use Parseval's theorem in order to convert the (continuous-time) inner product to a sum over Fourier coefficients. For example, in the case $1 \le i,k \le L$, we have
\begin{align}
M_{ik} &= \int_0^T h(t-t_i) h^*(t-t_k) dt \notag\\
       &= T\sum_{n\in\ZZ} \tilde{h}_n e^{-j 2\pi t_i n/T} \tilde{h}_n^* e^{j 2\pi t_k n/T} \notag\\
       &= T\sum_{n\in\ZZ} |\tilde{h}_n|^2 e^{-j 2\pi (t_k-t_i) n/T}, \quad 1 \le i,k \le L.
\end{align}
An analogous derivation can be carried out when $i$ or $k$ are in the complementary range $L+1, \ldots, 2L$.

Finally, having calculated the matrices $\J_\bth^{\mathrm{samp}}$ and $\M$, the CRB for sampled measurements is obtained using \eqref{eq:th:crb sampled}. We are now ready to compare this bound to the performance of practical estimators in some specific scenarios.

\subsection{Effect of the Pulse Shape}
\label{ss:pulse shape}

\begin{figure}
\subfigure[The pulse $g(t)$ is a filtered Dirac with $401$ Fourier coefficients.]
{
  \centerline{%
    \begin{overpic}{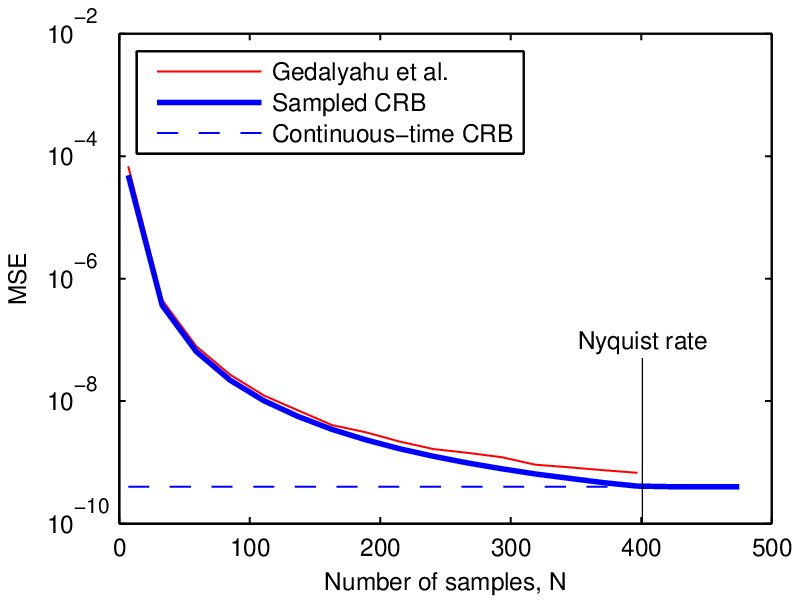}  
      \put(70,50){\includegraphics[scale=0.5]{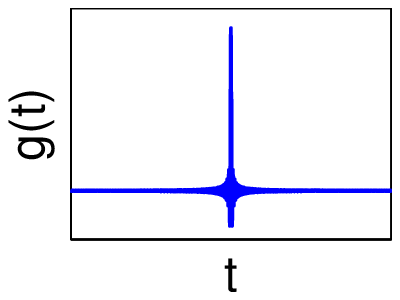}}  
    \end{overpic}
  }
  \label{fi:CRBvsN delta}
}
\\[1cm]
\subfigure[The pulse $g(t)$ contains $401$ nonzero Fourier coefficients which decrease monotonically with the frequency.]
{
  \centerline{%
    \begin{overpic}{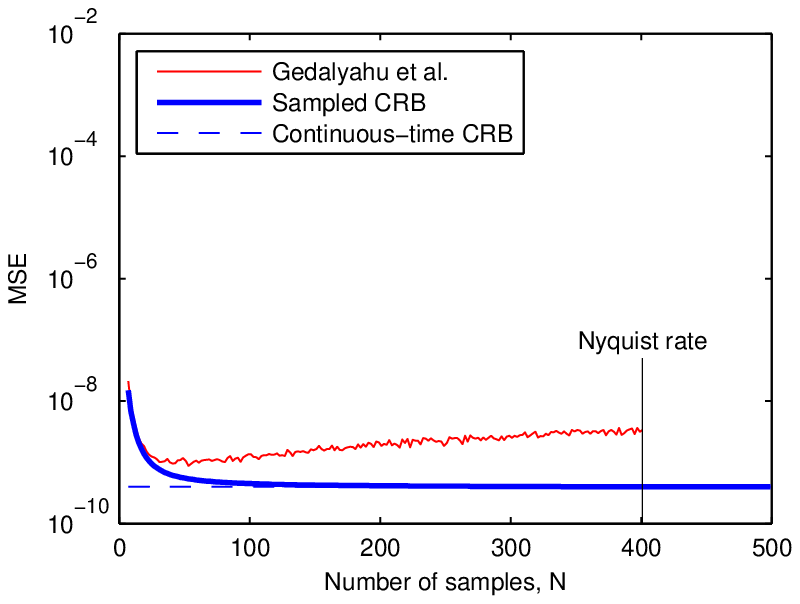}  
      \put(70,50){\includegraphics[scale=0.5]{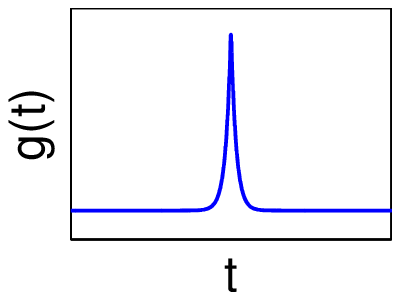}}  
    \end{overpic}
  }
  \label{fi:CRBvsN exp}
}
\\[1cm]
\subfigure[The pulse $g(t)$ is a filtered $\mathrm{rect}(\cdot)$ with $401$ Fourier coefficients.]
{
  \centerline{%
    \begin{overpic}{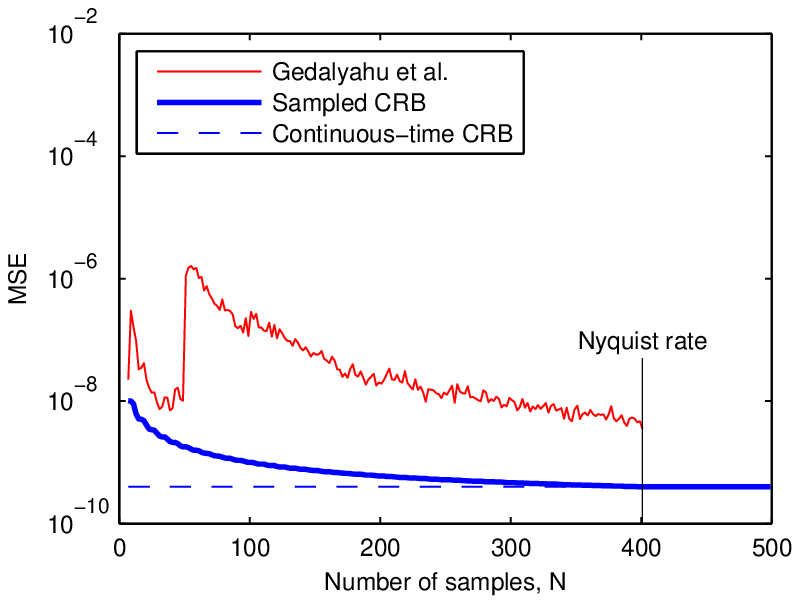}  
      \put(70,50){\includegraphics[scale=0.5]{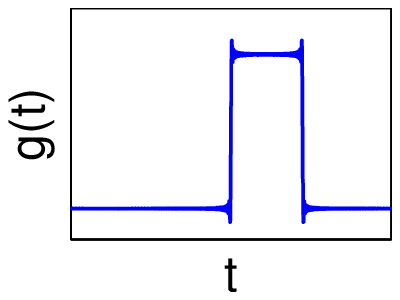}}  
    \end{overpic}
  }
  \label{fi:CRBvsN rect}
}
\caption{Comparison of the CRB and the performance of a practical estimator, as a function of the number of samples.}
\label{fi:CRBvsN}
\end{figure}

In Fig.~\ref{fi:CRBvsN}, we document several experiments comparing the CRB with the time-delay estimation technique of Gedalyahu et al.\ \cite{gedalyahu11}. Specifically, we sampled the signal $x(t)$ of \eqref{eq:per} using a set of exponential kernels, and used the matrix pencil method \cite{hua90} to estimate $x(t)$ from the resulting measurements. Since we are considering only continuous-time noise, applying an invertible linear transformation to the sampling kernels has no effect on our performance bounds (see Section~\ref{se:sampled}). The various kernels suggested in \cite{gedalyahu11} amount to precisely such an invertible linear transformation, and the same performance bound applies to all of these approaches. Moreover, under the continuous-time noise model, it can be shown that these techniques also exhibit the same performance. For the same reason, the performance reported here is also identical to the method of Vetterli et al.\ \cite{vetterli02}.

In our experiments, a signal containing $L=2$ pulses was constructed. The delays and amplitudes of the pulses were chosen randomly and are given by
\begin{align}
a_1 &= 0.3204, & t_1 &= 0.6678, \notag\\
a_2 &= 0.6063, & t_2 &= 0.9863.
\end{align}
Modifications of these parameters does not appear to significantly affect the reported results, except when the time delays are close to one another, a situation which will be discussed in depth in Section~\ref{ss:closely spaced}. The pulse $h(t)$ consisted of $|\KK|=401$ nonzero Fourier coefficients at positions $\KK = \{ -200, \ldots, 200 \}$. The CRB is plotted as a function of the number of samples $N$, where the sampling kernels are given by $s_n(t) = e^{j 2\pi nt/T}$ with $n \in \{-\lfloor N/2 \rfloor, \ldots, \lfloor N/2 \rfloor\}$. This is done because the matrix pencil method requires the sampling kernels to have contiguous frequencies.

In Fig.~\ref{fi:CRBvsN delta}, we chose $\tilde{h}_k=1$ for $-200 \le k \le 200$ and $\tilde{h}_k=0$ elsewhere; these are the low-frequency components of a Dirac delta function. The noise standard deviation was $\sigma_c = 10^{-5}$. In this case, for a fixed budget of $N$ samples, any choice of $N$ exponentials having frequencies in the range $-200 \le k \le 200$ is optimal according to the criterion of Section~\ref{se:bayesian}. As expected, the sampled CRB achieves the continuous-time bound $K \sigma_c^2$ when $N \ge |\KK|$. However, the CRB obtained at low sampling rates is higher by several orders of magnitude than the continuous-time limit. This indicates that the maxim of FRI theory, whereby sampling at the rate of innovation suffices for reconstruction, may not always hold in the presence of mild levels of noise. Indeed, if no noise is added in the present setting, then perfect recovery can be guaranteed using as few as $N=4$ samples; yet even in the presence of mild noise, our bounds demonstrate that performance is quite poor unless the number of samples is increased substantially. This result may provide an explanation for the previously observed numerical instability of FRI techniques \cite{vetterli02, tur11}.

As a further observation, we note that in this scenario, existing algorithms come very close to the CRB\@. Thus, the previously observed improvements achieved by oversampling are a result of fundamental limitations of low-rate sampling, rather than drawbacks of the specific technique used.

The same experiment is repeated in Fig.~\ref{fi:CRBvsN exp} with a pulse having Fourier coefficients $\tilde{h}_k = 1/(1 + 0.01 k^2)$. Since the Fourier coefficients decrease with $|k|$, in this case our choice of low-frequency sampling kernels is optimal. However, the SNR of the measurements $c_n$ decreases with $n$. As can be seen, this has a negative effect on the performance of the algorithm, which is not designed for high noise levels. Indeed, including low-SNR measurements causes the MSE not only to depart from the CRB, but eventually even to increase as more noisy samples are provided. In other words, one would do better to ignore the high-frequency measurements than to feed them to the recovery algorithm. Yet information is clearly present in these high-frequency samples, as indicated by the continual decrease of the CRB with increasing $N$. Thus, our analysis indicates that improved estimation techniques should be achievable in this case, in particular by careful utilization of low-SNR measurements.

The adverse effect of low-SNR measurements is exacerbated if, for a given $N$, one does not choose the $N$ largest Fourier coefficients. This is demonstrated in Fig.~\ref{fi:CRBvsN rect}. Here, the results of a similar experiment are plotted, in which $\tilde{h}_k = P \sinc(nP/T)$, $-200 \le k \le 200$. These are the $401$ lowest-frequency Fourier coefficients of a rectangular pulse having width $P$. In this case, the Fourier coefficients are no longer monotonically decreasing with $|k|$. Consequently, the sampling kernels $s_n(t) = e^{j 2\pi nt/T}$ with $n \in \{-\lfloor N/2 \rfloor, \ldots, \lfloor N/2 \rfloor\}$ do not correspond to the $N$ largest Fourier coefficients, and thus are not optimal. In particular, for the chosen parameters, $|\tilde{h}_{25}| = |\tilde{h}_{-25}|$ are considerably smaller than the rest of the coefficients. When $N \ge 50$, the corresponding measurements are included, causing the MSE to deteriorate significantly.

\subsection{Closely-Spaces Pulses}
\label{ss:closely spaced}

It is well-known that the estimation of pulse positions becomes ill-conditioned when several of the pulses are located close to one another. Intuitively, this is a consequence of the overlap between the pulses, which makes it more difficult to identify the precise location of each pulse. However, our goal is to estimate the signal $x(t)$ itself, rather than the positions of its constituent pulses. As we will see, for this purpose the effect of closely-spaced pulses is less clear-cut.

To study the effect of pulse position on the estimation error, we used a setup similar to the one of Fig.~\ref{fi:CRBvsN exp}, with the following differences. First, a higher noise level of $\sigma_c = 10^{-3}$ was chosen. Second, the signal consisted of $L=2$ pulses, with the first pulse at position $t_1=0.5$. The position of the second pulse was varied in the range $[0.3,0.7]$ to demonstrate the effect of pulse proximity on the performance. The setting was otherwise identical to that of Section~\ref{ss:pulse shape}. In particular, recall that $T=1$.

\begin{figure}
  \psfrag{t2}{{\small $t_2$}}
  \centerline{\includegraphics{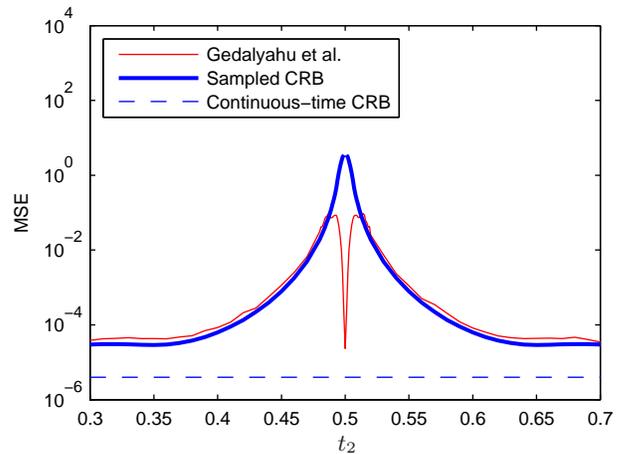}}  
  \caption{Comparison between the CRB and the performance of a practical estimator as a function of the pulse positions. The signal contains $L=2$ pulses, the first of which is located at $t_1=0.5$. The MSE is plotted as a function of the position of the second pulse.}
  \label{fi:t2}
\end{figure}

\begin{figure}
\subfigure[The spacing between the pulses is $0.04$.]
{\centerline{\includegraphics{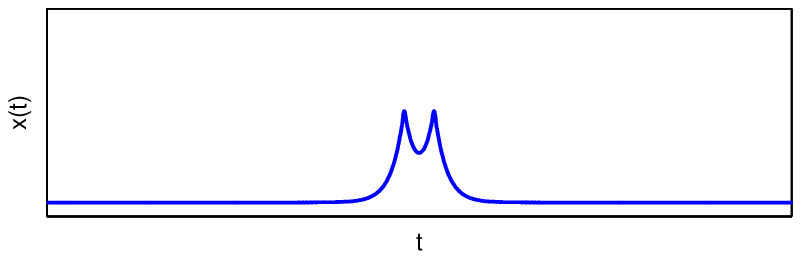}}
\label{fi:overlap 0.04}}  
\subfigure[The spacing between the pulses is $0.01$.]
{\centerline{\includegraphics{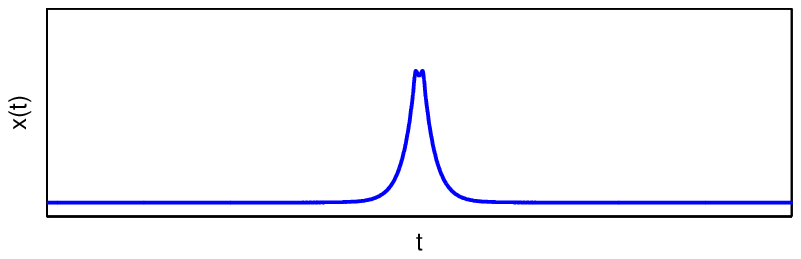}}
\label{fi:overlap 0.01}}  
\caption{Demonstration of the different levels of overlap between pulses.}
\label{fi:overlap}
\end{figure}

The results of this experiment are plotted in Fig.~\ref{fi:t2}, which documents both the values of the sampled CRB and the actual MSE obtained by the estimator of Gedalyahu et al.\ \cite{gedalyahu11}. The continuous-time CRB is also plotted, although, as is evident from Theorem~\ref{th:crb x from y}, this bound is a function only of the number of parameters determining the signal, and is therefore unaffected by the proximity of the pulses.

Several different effects are visible in Fig.~\ref{fi:t2}. First, as the two pulses begin to come closer, both the CRB and the observed MSE increase by several orders of magnitude; this occurs when $|t_1-t_2|$ is between about $0.15$ and $0.03$. (Of course, the precise distances at which these effects occur depend on the pulse width and other parameters of the experiment.) This level of proximity is demonstrated in Fig.~\ref{fi:overlap 0.04}. At this stage, the overlap between the pulses is sufficient to make it more difficult to estimate their positions accurately, but the separation between the pulses is still large, so that they are not mistaken for a single pulse.

As the pulses draw nearer each other, they begin to resemble a single pulse located at $(t_1+t_2)/2$ (see Fig.~\ref{fi:overlap 0.01}). Depending on the noise level, at some point the estimation algorithm will indeed identify the two pulses as one. Since our goal is to estimate $x(t)$ and not the pulse positions, such an ``error'' causes little deterioration in MSE\@. This is visible in Fig.~\ref{fi:t2} as the region in which the MSE of the practical algorithm ceases to deteriorate and ultimately decreases.

Interestingly, the CRB does not capture this improvement in performance. This failure is due to the fact that the CRB applies only to unbiased estimators, while the strategy utilized in \cite{gedalyahu11} becomes biased for closely-spaced pulses. For an estimator to be unbiased, it is required that the mean estimate, averaged over noise realizations, will converge to the true value of $x(t)$, which has a form similar to that of Fig.~\ref{fi:overlap 0.01}. The expectation of an estimator reconstructing a single pulse will not have the form of two closely-spaced pulses; such an estimator is thus necessarily biased. In other words, the discrepancy observed here results from the fact that in this case, biased techniques outperform the best unbiased approach.

\subsection{Non-Periodic and Semi-Periodic Signal Models}
\label{ss:nonper semiper}

As we have seen above, the reconstruction of signals of the form \eqref{eq:per} in the presence of noise is often severely hampered when sampled at or slightly above the rate of innovation. Rather than indicating a lack of appropriate algorithms, in many cases this phenomenon results from fundamental limits on the ability to recover such signals from noisy measurements. A similar effect was demonstrated \cite{tur11} in the non-periodic (or finite) pulse stream model \eqref{eq:nonper}. In fact, if one is allowed to sample a non-periodic pulse stream with arbitrary sampling kernels, then by designing kernels having sufficiently large time-domain support, one can capture all or most of the energy in the signal. This setting then essentially becomes equivalent to a periodic signal model \eqref{eq:per} in which the period is larger than the effective support of the pulse stream: One can imagine that the signal repeats itself beyond the sampled region, as this would not affect the measurements. Consequently, it is not surprising that the non-periodic model demonstrates substantial improvement in the presence of oversampling \cite{tur11}.

On the other hand, some types of FRI and union of subspace signals exhibit remarkable noise resilience, and do not appear to require substantial oversampling in the presence of noise \cite{gedalyahu09, mishali09}. As we now show, the CRB can be used to verify that such phenomena arise from a fundamental difference between families of FRI signals.

As an example, we compare the CRB for reconstructing the periodic signal \eqref{eq:per} with the semi-periodic signal \eqref{eq:semiper}. Recall that in the former case, each period consists of pulses having unknown amplitudes and time shifts. By contrast, in the latter signal, the time delays are identical throughout all periods, but the amplitudes can change from one period to the next.

While these are clearly different types of signals, an effort was made to form a fair comparison between the reconstruction capabilities in the two cases. To this end, we chose an identical pulse $g(t)$ in both cases. We selected the signal segment $[0,T_0]$, where $T_0 = 1$, and chose the signal parameters so as to guarantee an identical $T_0$-local rate of innovation. We also used identical sampling kernels in both settings: specifically, we chose the kernels \eqref{eq:fourier sampling} which measure the $N$ lowest frequency components of the signal.

\begin{figure}
  \centerline{\includegraphics{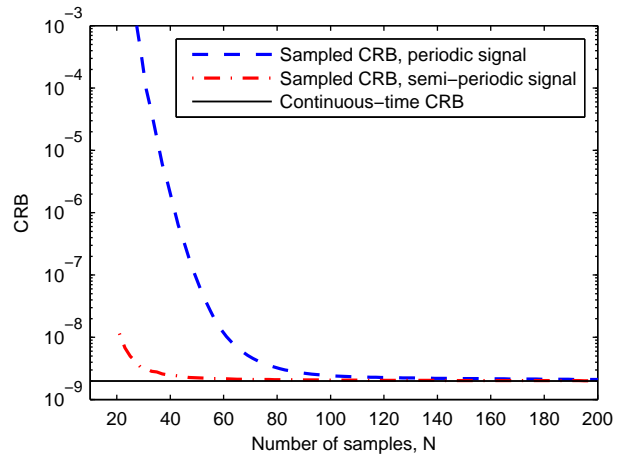}}  
  \caption{Comparison between the CRB for a periodic signal \eqref{eq:per} and a semi-periodic signal \eqref{eq:semiper}.}
  \label{fi:semiper}
\end{figure}

To simplify the analysis and focus on the fundamental differences between these settings, we will assume in this section that the pulses $g(t)$ are compactly supported, and that the time delays are chosen such that pulses from one period do not overlap with other periods. In other words, if the support of $g(t)$ is given by $[t_a, t_b]$, then we require
\beq \label{eq:no overlap req}
\min_\ell \{ t_\ell \} > t_a \quad \text{and} \quad \max_\ell \{ t_\ell \} < T - t_b.
\eeq
Specifically, we chose the pulse $g(t)$ used in Fig.~\ref{fi:CRBvsN exp}, which is compactly supported to a high approximation.

For the periodic signal, we chose $L=10$ pulses with random delays and amplitudes, picked so as to satisfy the condition \eqref{eq:no overlap req}. A period of $T=1$ was selected. This implies that the signal of interest is determined by $2L=20$ parameters ($L$ amplitudes and $L$ time delays).

To construct a semi-periodic signal with the same number of parameters, we chose a period of $T=1/9$ containing $L=2$ pulses. The segment $[0,T_0]$ then contains precisely $M=9$ periods, for a total of $20$ parameters. While it may seem plausible to require the same number of periods for both signals, this would actually disadvantage the periodic approach, as it would require the estimation of much more closely-spaced pulses.

The CRB for the periodic signal was computed as explained in Section~\ref{ss:compute CRB}, and the CRB for the semi-periodic signal can be calculated in a similar fashion. The results are compared with the continuous-time CRB in Fig.~\ref{fi:semiper}. Note that since the number of parameters to be estimated is identical in both signal models, the continuous-time CRB for the two settings coincides. Consequently, for a large number of measurements, the sampled bounds also converge to the same values. However, when the number of samples is closer to the rate of innovation, the bound on the reconstruction error for the semi-periodic signal is much lower than that of the periodic signal. As mentioned above, this is in agreement with previously reported findings for the two types of signals \cite{vetterli02, gedalyahu11, gedalyahu09}.

To find an explanation for this difference, it is helpful to recall that both signals can be described using the union of subspaces viewpoint (see Section~\ref{ss:union}). Each of the signals in this experiment is defined by precisely $20$ parameters, which determine the subspace to which the signal belongs and the position within this subspace. Specifically, the values of the time delays select the subspace, and the pulse amplitudes define a point within this subspace. Thus, in the above setting, the periodic signal contains $10$ parameters for selecting the subspace and $10$ additional parameters determining the position within it; whereas for the semi-periodic signal, only $2$ parameters determine the subspace while the remaining $18$ parameters set the location in the subspace. Evidently, identification of the subspace is challenging, especially in the presence of noise, but once the subspace is determined, the remaining parameters can be estimated using a simple linear operation (a projection onto the chosen subspace).  Consequently, if many of the unknown parameters identify the position within a subspace, estimation can be performed more accurately. This may provide an explanation for the difference between the two examined signal models.

As further evidence in support of this explanation, we recall from Section~\ref{ss:union} that the multiband signal model \eqref{eq:multiband} can also be viewed as a union of subspaces. Here, again, the parameters $\{ \omega_\ell \}_{\ell = 1}^L$ determining the subspace (i.e., the utilized frequency bands) are far fewer than the parameters $\{ a_\ell[n] \}$ selecting the point within the subspace (i.e., the content of each frequency band). In support of the proposed explanation, highly noise resistant algorithms can be constructed for the recovery of multiband signals \cite{mishali09, mishali09b}. An even more extreme case is the single subspace setting, exemplified by shift-invariant signals (Section~\ref{ss:sis}). In this case, all of the signal parameters are used to determine the position within the one possible subspace. As we have seen in Section~\ref{ss:subspace achieve crb}, in this case Nyquist-equivalent sampling at the rate of innovation achieves the continuous-time CRB\@.

\section{Conclusion}
\label{se:conclusion}

In this paper, we studied the inherent limitations in recovering FRI signals from noisy measurements. We derived a continuous-time CRB which provides a lower bound on the achievable MSE of any unbiased estimation method, regardless of the sampling mechanism. We showed that the rate of innovation $\rho_{T_0}$ is a lower bound on the ratio between the average MSE achievable by any unbiased estimator and the noise variance $\sigma_c^2$, \emph{regardless of the sampling method}. This stands in contrast to the noise-free interpretation of $\rho_{T_0}$ as the minimum sampling rate required for perfect recovery.

We next examined the CRB for estimating an FRI signal from a discrete set of noisy samples. We showed that the sampled bound is in general higher than the continuous-time CRB, and approaches it as the sampling rate increases. In general, the rate which is needed in order to achieve the continuous-time CRB is equal to the rate associated with the smallest subspace that encompasses all possible signal realizations. In particular, if a signal belongs to a union of subspaces, then the rate required to achieve the continuous-time bound is that associated with the sum of the subspaces. In some cases, this rate is finite, but in other cases the sum covers the entire space $L_2$ and no finite-rate technique achieves the CRB\@.

A consequence of these results is that oversampling can generally improve estimation performance. Indeed, our experiments demonstrate that sampling rates much higher than $\rho_{T_0}$ are required in certain settings in order to approach the optimal performance. Furthermore, these gains can be substantial: In some cases, oversampling can improve the MSE by several orders of magnitude. We showed that the CRB can be used to determine which estimation problems require substantial oversampling to achieve stable performance. As a rule of thumb, it appears that for union of subspace signals, performance is improved at low rates if most of the parameters identify the position within the subspace, rather than the subspace itself. Our analysis can also be used to identify cases in which no existing algorithm comes close to the CRB, implying that better approaches can be constructed. In particular, it seems that existing algorithms do not deal well with measurement sets having a wide dynamic range.

Lastly, we addressed the problem of choosing the sampling kernels. This was done by adopting a Bayesian framework, so that an optimality criterion can be rigorously defined. Using a generalization of the KLT, we showed that the optimal kernels are the eigenfunctions of the autocorrelation function of the signal. In the context of time-delay estimation, these kernels are exponentials with appropriately chosen frequencies. This choice coincides with recent FRI techniques \cite{gedalyahu09}.

\appendices

\section{Proof of Theorem~\ref{th:x from bth}}
\label{ap:th:x from bth}

The following notation will be used within this appendix. Let $\HH_1$ and $\HH_2$ be two measurable Hilbert spaces, and let $(\Omega, {\mathscr F}, P)$ be a probability space. Consider two random variables $u : \Omega \rightarrow \HH_1$ and $v : \Omega \rightarrow \HH_2$. Then, the notation $\E{uv^*}$ will be used to denote the linear operator $\HH_2 \rightarrow \HH_1$ such that, for any $h_1 \in \HH_1$ and $h_2 \in \HH_2$,
\beq \label{eq:def exp op}
\bra h_1, \E{uv^*} h_2 \ket_{\HH_1} = \E{ \bra h_1,u \ket_{\HH_1} \bra v,h_2 \ket_{\HH_2} }
\eeq
if the expectation exists for all $h_1$ and $h_2$.

We begin by stating two general lemmas which will be of use in the proof of Theorem~\ref{th:x from bth}.

\begin{lemma} \label{le:psd}
Let $\HH_1$ and $\HH_2$ be two Hilbert spaces, and consider the operators
\begin{align}
A : \HH_1 \rightarrow \HH_1, \notag\\
B : \HH_2 \rightarrow \HH_1, \notag\\
C : \HH_2 \rightarrow \HH_2.
\end{align}
Suppose $C$ is self-adjoint and invertible. Define the product Hilbert space $\HH_1 \times \HH_2$ in the usual manner, and suppose the operator $M : \HH_1\times\HH_2 \rightarrow \HH_1\times\HH_2$ defined by
\beq \label{eq:def M}
M \begin{pmatrix} h_1 \cr h_2 \end{pmatrix} = \begin{pmatrix} A h_1 + B h_2 \cr B^* h_1 + C h_2 \end{pmatrix}
\eeq
is positive semidefinite (psd). Then,
\beq \label{eq:le:psd}
A \succeq B C^{-1} B^*
\eeq
in the sense that the $\HH_1 \rightarrow \HH_1$ operator $A - BC^{-1}B^*$ is psd.
\end{lemma}

\begin{proof}
Since $M$ is psd, we have for any $h_1 \in \HH_1$ and $h_2 \in \HH_2$
\beq
\bra \begin{pmatrix} h_1 \cr h_2 \end{pmatrix} , M \begin{pmatrix} h_1 \cr h_2 \end{pmatrix} \ket_{\HH_1\times\HH_2} \ge 0
\eeq
which implies
\beq \label{eq:le:psd prf 1}
\bra h_1, Ah_1 \ket_{\HH_1} + 2 \Re\left[\bra h_1, B h_2 \ket_{\HH_2}\right] + \bra b, Cb \ket_{\HH_2} \ge 0.
\eeq
Choosing $h_2 = -C^{-1}B^*h_1$, we have that $\bra h_1, B h_2 \ket_{\HH_2} = -\bra B^* h_1 , C^{-1} B^* h_1 \ket_{\HH_1}$, which is real since $C^{-1}$ is self-adjoint. It follows from \eqref{eq:le:psd prf 1} that
\beq
\bra h_1, A h_1 \ket_{\HH_1} - \bra h_1, BC^{-1}B^*h_1 \ket_{\HH_1} \ge 0
\eeq
which leads to \eqref{eq:le:psd}, as required.
\end{proof}

\begin{lemma} \label{le:ext cauchy schwarz}
Let $\HH_1$ and $\HH_2$ be two Hilbert spaces and let $(\Omega, {\mathscr F}, P)$ be a probability space. Let $u : \Omega \rightarrow \HH_1$ and $v : \Omega \rightarrow \HH_2$ be random variables, and suppose the expectations $\E{uu^*}$, $\E{uv^*}$, and $\E{vv^*}$ exist as linear operators as defined in \eqref{eq:def exp op}. If $\E{vv^*}$ is invertible, then
\beq \label{eq:le:ext cauchy schwarz}
\E{uu^*} \succeq \E{uv^*} \left( \E{vv^*} \right)^{-1} \E{vu^*}.
\eeq
\end{lemma}

\begin{proof}
Let us denote $A = \E{uu^*}$, $B = \E{uv^*}$, and $C = \E{vv^*}$ and define the linear operator $M : \HH_1 \times \HH_2 \rightarrow \HH_1 \times \HH_2$ as in \eqref{eq:def M}. From \eqref{eq:def exp op}, for any $h_1 \in \HH_1$ and $h_2 \in \HH_2$ we have
\begin{align}
&\bra \begin{pmatrix} h_1 \cr h_2 \end{pmatrix} , M \begin{pmatrix} h_1 \cr h_2 \end{pmatrix} \ket_{\HH_1\times\HH_2} \notag\\
&= \bra h_1, Ah_1 \ket_{\HH_1} + 2 \Re\left[ \bra h_1, Bh_2 \ket_{\HH_1}\right] + \bra b, Cb \ket_{\HH_2} \notag\\
&= \mathbb{E}\Big\{ \left|\bra h_1,u \ket_{\HH_1}\right|^2 + 2\Re\left[\bra h_1,u \ket_{\HH_1} \bra v,h_2 \ket_{\HH_2}\right] \notag\\
&\hspace{10mm} + \left|\bra h_2,v \ket_{\HH_2}\right|^2 \Big\} \notag\\
&= \E{ \left| \bra h_1,u \ket_{\HH_1} + \bra v,h_2 \ket_{\HH_2} \right|^2 } \notag\\
&\ge 0.
\end{align}
Thus $M$ is a psd operator. Invoking Lemma~\ref{le:psd} yields \eqref{eq:le:ext cauchy schwarz}, as required.
\end{proof}

We are now ready to prove Theorem~\ref{th:x from bth}.

\begin{proof}[Proof of Theorem~\ref{th:x from bth}]
Throughout the proof, let $\bth$ be a fixed parameter and consider all functions as implicitly dependent on $\bth$. Define the random variables
\begin{align}
u: \Omega &\rightarrow \HH:    & u(\omega) &= \hat{x}(y(\omega)) - h_\bth, \\
v: \Omega &\rightarrow \RR^K:  & v(\omega) &= \pd{\log p_\bth(y(\omega))}{\bth}.
\end{align}
We then have the linear operators $\E{vv^*}: \RR^K \rightarrow \RR^K$, $\E{uu^*}: \HH \rightarrow \HH$, and $\E{uv^*}: \RR^K \rightarrow \HH$, which satisfy
\begin{align}
\E{vv^*} &= \J_\bth, \label{eq:vv*=J}\\
\bra \varphi_i, \E{uu^*} \varphi_j \ket &= \E{ \bra \varphi_i,u \ket \bra u,\varphi_j \ket }, \\
\bra \varphi_i, \E{uv^*} \e_j \ket &= \E{ \bra \varphi_i,u \ket \pd{\log p_\bth(y)}{\theta_j} },
\end{align}
where $\{ \varphi_n \}_{n\in\ZZ}$ denotes a complete orthonormal basis for $\HH$. The operator $\E{uu^*}$ can be thought of as the covariance of $\hat{x}$, and is well-defined since, by \eqref{eq:finite energy}, $\hat{x}$ has finite variance. Indeed, we have
\beq
\sum_{i\in\ZZ} \bra \varphi_i, \E{uu^*} \varphi_i \ket = \E{\|\hat{x}-h_\bth\|_{L_2}^2} < \infty
\eeq
so that $\E{uu^*}$ is not only well-defined, but a trace class operator. Furthermore, $\E{vv^*} = \J_\bth$ is well-defined and invertible by Assumption~P\ref{ass:fim}. The operator $\E{uv^*}$ is thus also well-defined by virtue of the Cauchy--Schwarz inequality.

To prove the theorem, we will show that
\beq \label{eq:ops are equal}
\E{uv^*} = \pdht
\eeq
and then obtain the required result by applying Lemma~\ref{le:ext cauchy schwarz}. To demonstrate \eqref{eq:ops are equal}, observe that
\begin{align} \label{eq:gen crb prf 1}
& \bra \varphi_i, \E{uv^*} \e_j \ket
 = \E{ \bra \varphi_i,u \ket \pd{\log p_\bth(y)}{\theta_j} } \notag\\
&= \int \bra \varphi_i, \hat{x}(y)-h_\bth \ket \frac{1}{p(y;\bth)} \pd{p(y;\bth)}{\theta_j} p(y;\bth) \Pdy \notag\\
&= \int \bra \varphi_i, \hat{x}(y)-h_\bth \ket \lim_{\Delta \rightarrow 0} \frac{p(y;\bth+\Delta \e_j) - p(y;\bth)}{\Delta} \Pdy .
\end{align}
By Assumption~P\ref{ass:dominating func}, for any sufficiently small $\Delta > 0$ we have
\begin{align} \label{eq:dom conv}
&\left| \bra \varphi_i, \hat{x}(y)-h_\bth \ket \frac{p(y;\bth+\Delta \e_j) - p(y;\bth)}{\Delta} \right| \notag\\
&\quad\le \left| \bra \varphi_i, \hat{x}(y)-h_\bth \ket \right| q(y,\bth).
\end{align}
Let us demonstrate that the right-hand side of \eqref{eq:dom conv} is absolutely integrable. By the Cauchy--Schwarz inequality,
\begin{align} \label{eq:dom conv 2}
&\left( \int \left| \bra \varphi_i, \hat{x}(y)-h_\bth \ket \right| q(y,\bth) \Pdy \right)^2 \notag\\
&\quad\le \int \left|\bra \varphi_i, \hat{x}(y)-h_\bth \ket\right|^2 \Pdy \cdot \int q^2(y,\bth) \Pdy.
\end{align}
The rightmost integral in \eqref{eq:dom conv 2} is finite by virtue of \eqref{eq:q sq int}. As for the remaining integral, we have
\begin{align}
\int &\left|\bra \varphi_i, \hat{x}(y)-h_\bth \ket\right|^2 \Pdy \notag\\
& \le \int \|\hat{x}(y)-h_\bth\|^2 \Pdy \notag\\
& \stackrel{\text{(a)}}{\le} \int \left( \|\hat{x}(y)\| + \|h_\bth\| \right)^2 \Pdy \notag\\
& \stackrel{\text{(b)}}{\le} \int \|\hat{x}(y)\|^2 \Pdy + \|h_\bth\|^2 \int \Pdy \notag\\
&\qquad{}+ 2 \|h_\bth\| \left( \int \|\hat{x}(y)\|^2 \Pdy \right)^{1/2} \notag\\
&\stackrel{\text{(c)}}{<} \infty
\end{align}
where we have used the triangle inequality in (a), the Cauchy--Schwarz inequality in (b), and the assumption \eqref{eq:finite energy} that $\hat{x}$ has finite energy in (c). We conclude that \eqref{eq:dom conv} is bounded by an absolutely integrable function, and we can thus apply the dominated convergence theorem to \eqref{eq:gen crb prf 1}, obtaining
\begin{align} \label{eq:gen crb prf 2}
\bra \varphi_i, \E{uv^*} \e_j \ket &= \pd{}{\theta_j} \int \bra \varphi_i, \hat{x}(y) \ket p(y;\bth) \Pdy \notag\\
&\quad {} - \bra \varphi_i, h_\bth \ket \pd{}{\theta_j} \int p(y;\bth) \Pdy.
\end{align}
The second integral in \eqref{eq:gen crb prf 2} equals 1 and its derivative is therefore 0. Thus we have
\begin{align} \label{eq:gen crb prf 3}
\bra \varphi_i, \E{uv^*} \e_j \ket
&= \pd{ \E{ \bra \varphi_i, \hat{x}(y) \ket } }{\theta_j}
 = \pd{ \bra \varphi_i, h_\bth \ket }{\theta_j}.
\end{align}

On the other hand, note that the Fr\'echet derivative $\partial h_\bth / \partial \bth$ of \eqref{eq:pdht frechet} coincides with the G\^ateaux derivative of $h_\bth$. In other words, for any vector $\bv \in \RR^K$, we have
\beq \label{eq:pdht gateaux}
\pdht\bv = \lim_{\eps \rightarrow 0} \frac{h_{\bth + \eps\bv} - h_\bth}{\eps}.
\eeq
It follows that
\beq \label{eq:pdht gateaux 2}
\bra \varphi_i, \pdht \e_j \ket = \pd{ \bra \varphi_i, h_\bth \ket }{\theta_j}.
\eeq
Since $\E{uv^*}$ and $(\partial h_\bth/\partial \bth)^*$ are both linear operators, \eqref{eq:gen crb prf 3} and \eqref{eq:pdht gateaux 2} imply that the two operators are equal, demonstrating \eqref{eq:ops are equal}. Applying Lemma~\ref{le:ext cauchy schwarz} and using the results \eqref{eq:vv*=J} and \eqref{eq:ops are equal}, we have
\beq \label{eq:gen crb prf 4}
\E{(\hat{x}-h_\bth)(\hat{x}-h_\bth)^*} \succeq \left( \pdht \right) \J_\bth^{-1} \left( \pdht \right)^*.
\eeq
As we have seen, the left-hand side of \eqref{eq:gen crb prf 4} is trace class, and thus so is the right-hand side. Taking the trace of both sides of the equation, we obtain
\beq
\E{\|\hat{x}-h_\bth\|^2} \ge \Tr \left( \left( \pdht \right) \J_\bth^{-1} \left( \pdht \right)^* \right)
\eeq
which is equivalent to \eqref{eq:th:x from bth}, as required.
\end{proof}

\section{Maximization of \eqref{eq:MSE_sub_vn}}
\label{ap:max_MSE_sub_vn}

The task of maximizing \eqref{eq:MSE_sub_vn} is most easily accomplished by optimizing the coordinates of $s_n(t)$ in the orthonormal basis for $L_2[0,T]$ generated by the eigenfunctions of $R_X(t,\eta)$. Specifically, the function  $s_n(t)$ can be written as
\begin{align} \label{eq:sn}
s_n(t)=\sum_{k=1}^\infty \alpha_k^n \left(\lambda_k+\sigma_c^2\right)^{-\frac{1}{2}}\psi_k(t),
\end{align}
with $\{\psi_k(t)\}_{k=1}^\infty$ and $\{\lambda_k\}_{k=1}^\infty$ of \eqref{eq:eigenfunctions}. (The coefficients $(\lambda_k+\sigma_c^2)^{-1/2}$ are inserted since they simplify the subsequent analysis.) Now, by Mercer's theorem, $R_X(t,\eta)$ can be expressed as
\begin{equation}
R_X(t,\eta) = \sum_{\ell=1}^\infty \lambda_\ell \psi_\ell(t) \psi_\ell^*(\eta),
\end{equation}
where the convergence is absolute and uniform. Therefore
\begin{align} \label{eq:Rx sn}
\int_0^T \!\! &R_{X}(t,\tau)s_n(\tau)d\tau \nonumber\\
&= \int_0^T \!\!\sum_{k=1}^\infty \alpha_k^n \left(\lambda_k+\sigma_c^2\right)^{-\frac{1}{2}} \psi_k(\tau)\sum_{\ell=1}^\infty \lambda_\ell \psi_\ell(t) \psi_\ell^*(\tau) d\tau \nonumber\\
&= \sum_{k=1}^\infty \alpha_k^n \frac{\lambda_k}{\left(\lambda_k+\sigma_c^2\right)^{\frac{1}{2}}} \psi_k(t),
\end{align}
and consequently, by Parseval's theorem, \eqref{eq:MSE_sub_vn} is given by
\begin{align}\label{eq:num}
&\sum_{n=1}^N\int_0^T\left|\int_0^T  R_{X}(t,\tau)s_n(\tau)d\tau\right|^2dt  \notag\\
&\hspace{10mm}= \sum_{n=1}^N\int_0^T\left| \sum_{k=1}^\infty \alpha_k^n \frac{\lambda_k}{\left(\lambda_k+\sigma_c^2\right)^{\frac{1}{2}}} \psi_k(t)  \right|^2dt \notag\\
&\hspace{10mm}= \sum_{n=1}^N\sum_{k=1}^\infty |\alpha_k^n|^2 \frac{\lambda_k^2}{\lambda_k+\sigma_c^2}.
\end{align}
Similarly, using \eqref{eq:sn} and \eqref{eq:Rx sn}, we have
\begin{align} \label{eq:sm Rx sn}
\iint_{0}^{T} s_m^*(t) R_{X}(t,\tau) s_n(\tau) d\tau dt = \sum_{k=1}^\infty \alpha_k^n(\alpha_k^m)^* \frac{\lambda_k}{\lambda_k+\sigma_c^2}
\end{align}
and by \eqref{eq:sn}
\begin{align} \label{eq:sm sn}
\sigma_c^2 \int_{0}^{T} s_m^*(t) s_n(\tau) dt = \sum_{k=1}^\infty \alpha_k^n(\alpha_k^m)^* \frac{\sigma_c^2}{\lambda_k+\sigma_c^2}.
\end{align}

Combining \eqref{eq:sm Rx sn} and \eqref{eq:sm sn}, the set of constraints \eqref{eq:s_const} is translated to
\begin{align}\label{eq:s_con_alpha}
\sum_{k=1}^\infty \alpha_k^m(\alpha_k^n)^*=\delta_{m,n}
\end{align}
for every $m,n=1,\ldots,N$. Consequently, our problem has now been reduced to
\beq \label{eq:opt alpha}
\max_{\{\alpha_k^n\}} \sum_{n=1}^N\sum_{k=1}^\infty |\alpha_k^n|^2 \frac{\lambda_k^2}{\lambda_k+\sigma_c^2}
\quad \text{s.t. } \sum_{k=1}^\infty \alpha_k^m(\alpha_k^n)^*=\delta_{m,n}.
\eeq

We now show that the sequences $\{\alpha_k^n\}$ which solve \eqref{eq:opt alpha} must satisfy $\alpha_k^n=0$ for every $k>N$ and $n=1,\ldots,N$. To see this, assume to the contrary that the $n$th sequence satisfies $\alpha_\ell^n\neq0$ for some $\ell>N$. We can then replace this sequence by a sequence $\{\tilde{\alpha}_k^n\}_{k\in\ZZ}$ satisfying
\begin{equation}
|\tilde{\alpha}_k^n|^2=\begin{cases}|\alpha_k^n|^2+a_k^2 & 1\leq k\leq N \\
                                    0                    & k=\ell \\
                                    |\alpha_k^n|^2       & N<k \text{ and } k\neq\ell\end{cases}
\end{equation}
where $\sum_{k=1}^Na_k^2=|\alpha_\ell^n|^2$ (to ensure that $\sum_{k\in\ZZ}|\tilde{\alpha}_k^n|^2=1$). Such a set of coefficients $\{a_k\}_{k=1}^N$ can always be found since the $N$-term truncation of the remaining $N-1$ sequences cannot span $\CC^N$. With this sequence, the $n$th summand in the objective of \eqref{eq:opt alpha} becomes
\begin{align}
\sum_{k=1}^\infty |\tilde{\alpha}_k^n|^2 &\frac{\lambda_k^2}{\lambda_k + \sigma_c^2} = \sum_{k=1}^\infty |\alpha_k^n|^2 \frac{\lambda_k^2}{\lambda_k+\sigma_c^2} \nonumber\\
 &\hspace{1.75cm}+\left(\sum_{k=1}^N a_k^2 \frac{\lambda_k^2}{\lambda_k+\sigma_c^2} - |\alpha_\ell^n|^2\frac{\lambda_\ell^2}{\lambda_\ell+\sigma_c^2}\right)\nonumber\\
&\geq \sum_{k=1}^\infty |\alpha_k^n|^2 \frac{\lambda_k^2}{\lambda_k+\sigma_c^2} + \frac{\lambda_\ell^2}{\lambda_\ell+\sigma_c^2}\left(\sum_{k=1}^N a_k^2  - |\alpha_\ell^n|^2\right) \nonumber\\
&= \sum_{k=1}^\infty |\alpha_k^n|^2 \frac{\lambda_k^2}{\lambda_k+\sigma_c^2},
\end{align}
where we used the fact that $\lambda_k\geq\lambda_{\ell}$ for every $k<\ell$ and that $z^2/(a+z)$ is a monotone increasing function of $z$ for all $z>0$. This contradicts the optimality of $\{\alpha_k^n\}_{k\in\ZZ}$. Therefore, the set of sequences maximizing \eqref{eq:opt alpha} satisfy $\alpha_k^n=0$ for every $k>N$ and $n=1,\ldots,N$.

It remains to determine the optimal values of the first $N$ elements of each of the $N$ sequences $\{\alpha_k^n\}_{k\in\ZZ}$, $n=1,\ldots,N$. For this purpose, let $\A$ denote the $N\times N$ matrix whose entries are $\A_{k,n}=\alpha_k^n$ and let $\bLambda$ be a diagonal matrix with $\bLambda_{k,k}=\lambda_k^2/(\lambda_k+\sigma_c^2)$. Then, the constraint \eqref{eq:s_con_alpha} can be written as $\A^*\A=\I_N$, which is equivalent to $\A\A^*=\I_N$. Now, the objective in \eqref{eq:opt alpha} can be expressed as
\begin{align}
\sum_{n=1}^N\sum_{k=1}^\infty |\alpha_k^n|^2 \frac{\lambda_k^2}{\lambda_k+\sigma_c^2} &= \Tr\{\A^*\bLambda\A\}\nonumber\\
 &= \Tr\{\A\A^*\bLambda\} \nonumber\\
 &= \Tr\{\bLambda\},
\end{align}
which is independent of $\A$. Therefore, we conclude that any set of orthonormal sequences $\{\alpha_k^n\}_{k\in\ZZ}$, $n=1,\ldots,N$, whose elements vanish for every $k>N$ is optimal.

\bibliographystyle{IEEEtran}
\bibliography{IEEEabrv,zvika}

\end{document}